\documentclass[acmsmall]{acmart}

\AtBeginDocument{%
  }

\usepackage{amsthm, amsmath, mathtools}
\usepackage{upgreek}
\usepackage{amsfonts}
\usepackage{enumitem}
\usepackage[normalem]{ulem}
\usepackage[linesnumbered,vlined,ruled,commentsnumbered]{algorithm2e}
\usepackage[para,online,flushleft]{threeparttable}
\usepackage{multirow}
\usepackage{multicol}
\usepackage{hyperref}
\usepackage{url}
\usepackage{bm}
\usepackage[labelfont=bf,skip=0pt]{caption,subfig}
\usepackage{flushend}
\usepackage{tabularx}
\usepackage[english]{babel}
\usepackage{graphicx}
\graphicspath{{./figs/}}
\usepackage{thmtools}
\usepackage{listings}
\usepackage{tablefootnote}
\usepackage{enumitem}
\graphicspath{{DurableJoin/imgs/}}
\newcommand{\CT}{\mathsf{CT}}
\newcommand{\RT}{\mathsf{RT}}
\newcommand{\IT}{\mathsf{IT}}
\renewcommand{\paragraph}[1]{\smallskip \noindent {\bf #1}}
\newcommand{\OUT}{\textsf{OUT}}
\newcommand{\OUTOFF}{\textsf{OUT}}

\newcommand{\OUTS}{\textsf{OUT}}

\newcommand{\C}{\mathcal{C}}

\newcommand{\T}{\mathcal{T}}

\newcommand{\I}{\mathcal{I}}

\renewcommand{\S}{\mathcal{S}}

\newcommand{\polylog}{\mathop{\mathrm{polylog}}}
\renewcommand{\O}{\Tilde{O}}

\newcommand{\eat}[1]{}

\newcommand{\change}[1]{{\color{black} #1}}
\newcommand{\newchange}[1]{{\color{black} #1}}
\DeclareMathOperator*{\argmax}{arg\,max}

\def\mparagraph#1{\par\smallskip\noindent{\textbf{#1.}}\quad\parindent 1.5em}
\newcommand{\eps}{\ensuremath{\varepsilon}}
\renewcommand{\Re}{\ensuremath{\mathbb{R}}}
\newcommand{\DTR}{\ensuremath{\mathcal{D}_R}}
\newcommand{\DTRQ}{\ensuremath{\mathcal{Q}_R}}
\newcommand{\DTC}{\ensuremath{\mathcal{D}}}
\newcommand{\DTCQ}{\mathsf{durableBallQ}}
\newcommand{\DTCE}{\ensuremath{\mathcal{D}'}}
\newcommand{\DTCEQ}{\mathsf{durableBallQ}'}
\newcommand{\dist}{\ensuremath{\phi}}
\newcommand{\ddim}{\ensuremath{\rho}}

\newcommand{\problemone}{\textsf{DurableTriangle}}
\newcommand{\problemtwo}{\textsf{IncrDurableTriangle}}
\newcommand{\problemthree}{\textsf{AggDurablePair}}
\newcommand{\SUM}{\textsf{SUM}}
\newcommand{\UNION}{\textsf{UNION}}
\newcommand{\rep}{\textsf{Rep}}
\renewcommand{\S}{\ensuremath{\mathcal{S}}}
\newcommand{\ITSUM}{\ensuremath{\IT^\Sigma}}
\newcommand{\DTCSUM}{\ensuremath{\smash{\DTC^\Sigma}}}
\newcommand{\NewerC}{\ensuremath{\Lambda}}
\newcommand{\OlderC}{\ensuremath{\smash{\overline{\Lambda}}}}
\newcommand{\Smax}{\ensuremath{\mathcal{S}_\alpha}}
\newcommand{\Scur}{\ensuremath{\mathcal{S}_\beta}}
\newcommand{\ITUNION}{\ensuremath{\IT^\cup}}
\newcommand{\DTCUNION}{\ensuremath{\smash{\DTC^\cup}}}

\begin{document}
\title{On Reporting Durable Patterns in Temporal Proximity Graphs}

\author{Pankaj K. Agarwal}
\affiliation{%
  \institution{Duke University}
  \city{Durham}
  \country{USA}}
\email{pankaj@cs.duke.edu}

\author{Xiao Hu}
\affiliation{%
  \institution{University of Waterloo}
  \city{Waterloo}
  \country{Canada}}
\email{xiaohu@uwaterloo.ca}

\author{Stavros Sintos}
\affiliation{%
  \institution{University of Illinois at Chicago}
  \city{Chicago}
  \country{USA}}
\email{stavros@uic.edu}

\author{Jun Yang}
\affiliation{%
  \institution{Duke University}
  \city{Durham}
  \country{USA}}
\email{junyang@cs.duke.edu}

\setcopyright{acmlicensed}
\acmJournal{PACMMOD}
\acmYear{2024} \acmVolume{2} \acmNumber{2 (PODS)} \acmArticle{81} \acmMonth{5}\acmDOI{10.1145/3651144}

\begin{abstract}
Finding patterns in graphs is a fundamental problem in databases and data mining.
In many applications, graphs are {\em temporal} and evolve over time,
so we are interested in finding {\em durable patterns}, such as triangles and paths, which persist over a long time.
While there has been work on finding durable simple patterns, existing algorithms do not have provable guarantees and run in strictly super-linear time.
The paper leverages the observation that many graphs arising in practice are naturally \emph{proximity graphs} or can be approximated as such,
where nodes are embedded as points in some high-dimensional space, and two nodes are connected by an edge if they are close to each other.
We work with an implicit representation of the proximity graph, where nodes are additionally annotated by time intervals,
and design near-linear-time algorithms for finding (approximately) durable patterns above a given durability threshold.
We also consider an interactive setting where a client experiments with different durability thresholds in a sequence of queries;
we show how to compute incremental changes to result patterns efficiently in time near-linear to the size of the changes.
\end{abstract}




\received{December 2023}
\received[revised]{February 2024}
\received[accepted]{March 2024}

\maketitle

\section{Introduction}

Finding patterns in large graphs is a fundamental problem in databases and data mining.
In many practical applications, graphs evolve over time, and we are often more interested in patterns that are ``durable,'' i.e., persisting over a long time.
Here are two examples of finding durable patterns in temporal graphs.

\begin{example}
\label{exmpl:1}
Consider an online forum with social networking features, where users with similar profiles are connected as friends.
Each user may be active on the forum only for a period of time during the day.
We are interested in finding cliques of connected users who are simultaneously active for a sufficiently long time period.
Such queries are useful 
to forum administrators who want to understand how the social network influences user interactions and leverage this knowledge to promote more interactions.
\end{example}

\begin{example}
\label{exmpl:2}
Consider a co-authorship graph where two researchers are connected if they have written at least one paper together.
Further suppose that researchers are each associated with a time period when they remain active in research.
Besides researchers with direct co-authorship, we might be interested in pairs who co-authored with a set of common researchers over a long period of time.
\newchange{We would not be interested in researchers with a common co-author if the respective collaborations happened at distant times.}
\end{example}

The problem of finding (durable) patterns\change{, such as triangles,} in general graphs is \change{challenging}: known (conditional) lower bounds suggest that it is unlikely to have near-linear algorithms~\cite{abboud2014popular, alon2008testing, patrascu2010towards}.
However, many graphs that arise in practice are naturally \emph{proximity graphs}, or can be approximated as such.
In proximity graphs, nodes are embedded as points in some high-dimensional space, and two nodes are connected by an edge if they are close to each other (i.e., their distance is within some threshold).
For example, social networks such as 
Example~\ref{exmpl:1} can be embedded (with small error) in the space of user profiles with a low intrinsic dimension~\cite{verbeek2014metric}.
Similarly, in the co-authorship graph, where two authors nodes are connected if they have written at least $m$ papers together for some $m\geq 1$,
the nodes can also be embedded in a space with low intrinsic dimension~\cite{yang2012defining}.
Generally, for many graphs arising in a wide range of applications (e.g. social network, transportation network, Internet), there exist appropriate node embeddings that preserve the structures and shortest paths in the original graphs~\cite{verbeek2014metric, zhao2010orion, zhao2011efficient}.
This observation enables us to leverage the properties of proximity graphs to develop efficient algorithms for finding patterns in such graphs, which overcome the hardness of the problem on arbitrary graphs.

This paper hence tackles the problem of finding durable patterns in temporal \emph{proximity} graphs, for which we are not aware of efficient algorithms.
For simplicity, we assume in this paper that the embedding of the graph is given -- there are efficient algorithms for computing graph embeddings~\cite{verbeek2014metric, chen1997algorithmic, kutuzov2020making, cai2018comprehensive, zeng2019accurate}.
We work with an implicit representation of the proximity graphs -- nodes represented as points and edges defined between pairs of points within a threshold distance in the embedding space. We never construct the graph itself explicitly.
We design efficient algorithms whose running time depend on the number of nodes and the intrinsic dimension (\emph{doubling dimension}) of the data.
Our approach extends naturally to other classes of graphs including interval graphs, permutation graphs, and grid graphs.
%
Next, we formally define the problems we study.
Our notation is summarized in Table~\ref{Table:Notation}.

\subsection{Problem Definitions}
\label{sec:intro:defs}

Let $(P,\dist)$ be a metric space over a set of $n$ points $P\subset \Re^d$, for some $d\geq 1$, and a metric $\dist$.
For a parameter $r>0$, let $G_\dist(P,r)=(P,E)$ where $E=\{(p,q) \mid \dist(p,q)\leq r\}$ be a {\em proximity graph}, also called a {\em unit disk} graph.
For simplicity, we assume $r=1$ and let $G_\dist(P)=G_\dist(P,1)$.
Suppose a function $I$ assigns each point $p\in P$ to a time interval called its \emph{lifespan}, denoted $I_p=[I_p^-, I_p^+]$.
We can interpret the lifespan of $p$ as inserting $p$ at time-stamp $I_p^-$ and deleting it at time-stamp $I_p^+$.
We use $(P,\dist,I)$ to refer to the
\emph{temporal} proximity graph,
or the underlying metric space with points annotated with interval lifespans.
For simplicity, all defined problems assume that the query pattern is triangle. As we point out in Section~\ref{sec:intro:results}, all techniques are extended to more general patterns.
\change{For an interval $I$, we define $|I|$ as the length of $I$. If $I$ is a set of intervals then $|I|$ is defined as the length of the union of intervals in $I$}.

\mparagraph{Durable triangles}
%
A triplet $(p_1, p_2, p_3) \in P \times P \times P$ forms a {\em triangle} in $G_\dist(P)$ if $\dist(p_1, p_2)$, $\dist(p_2, p_3)$, $\dist(p_1,p_3)\le 1$.
We also introduce an approximate notion of triangles:
for a parameter $\eps>0$, a triplet $(p_1, p_2, p_3)$ \newchange{forms an {\em $(1+\eps)$-approximate triangle}, or {\em $\eps$-triangle} for brevity}, if $\dist(p_1, p_2)$, $\dist(p_2, p_3)$, $\dist(p_1,p_3)\leq 1+\eps$.
The {\em lifespan} of $(p_1, p_2, p_3)$ is defined as $I(p_1, p_2, p_3) = I_{p_1}\cap I_{p_2}\cap I_{p_3}$.
For a durability parameter $\tau > 0$, $(p_1, p_2, p_3)$ is \emph{$\tau$-durable} if $|I(p_1, p_2, p_3)|\ge \tau$.
Let $T_{\tau}, T_{\tau}^{\eps}$ be the set of $\tau$-durable triangles, and $\tau$-durable $\eps$-triangles respectively. Note that $T_{\tau}\subseteq T_{\tau}^{\eps}$.
\newchange{Given a $\tau$-durable triangle with three points, the point that \emph{anchors} the triangle is the one whose lifespan starts the latest among the three. By convention, we will list the anchor first in the triplet; i.e., in a $\tau$-durable triangle $(p, q, s)$, we have $I_p^-\geq \max\{I_q^-, I_s^i\}$.}

\begin{table}[t]
\centering
 \begin{tabular}{c|c} 
 \hline
 $P$ & point set\\ \hline
 $n$ & $|P|$\\ \hline
 $\phi$ & distance function\\ \hline
 $\ddim$ & doubling dimension\\ \hline
 $\eps$ & distance approximation \\ \hline
  $\tau$ & durability parameter\\ \hline
 $I_p=[I_p^-,I_p^+]$ & lifespan (interval) of point $p$\\ \hline
 $T_\tau$ & $\tau$-durable triangles\\ \hline
 $T_\tau^\eps$ & $\tau$-durable $\eps$-triangles\\ \hline
 $K_\tau^\eps$ & $\tau$-SUM durable $\eps$-pairs\\ \hline
 $K_{\tau, \kappa}^\eps$ & $(\tau,\kappa)$-UNION durable $\eps$-pairs\\ \hline
 $\OUTOFF$ & Output size\\ \hline
\end{tabular}
\vspace{1em}
\caption{Table of Notations}
\label{Table:Notation}
\end{table}

\begin{definition}[\textbf{\problemone}]
    Given $(P,\dist, I)$ and $\tau \ge 0$, it asks to report all $\tau$-durable triangles (or $\eps$-triangles).   
\end{definition}
%

Suppose we have embedded the social network in Example~\ref{exmpl:1} as a proximity graph where nodes represent users.
The goal is to find triplets (or generally cliques) of users who are simultaneously active on the forum.%
\footnote{While for simplicity of exposition we assume that each node has a single-interval lifespan, it is straightforward to extend our temporal model consider multiple-interval lifespans, with the complexities of our solutions in the following sections increased by a factor equal to the maximum number of intervals per lifespan.}

In some use cases, we do not have a clear choice of the durability parameter $\tau$ in mind, and we may want to explore with different settings.
Supporting this mode of querying motivates the problem of incrementally reporting $\tau$-durable triangles.
Here, queries arrive in an online fashion, each specifying a different durability parameter $\tau_1, \tau_2, \ldots$.
Instead of computing each query $\tau_{i+1}$ from scratch, we want to leverage the previous query result $T_{\tau_i}$ and only incrementally compute what is new.
Note that every $\tau$-durable triangle must also be $\tau'$-durable for every $\tau' \le \tau$.
Therefore, if $\tau_{i+1} \ge \tau_i$, we have $T_{\tau_{i+1}} \subseteq T_{\tau_i}$ so we simply need to filter the old results to obtain new ones (assuming we remember results together with their lifespans). The more interesting case is when $\tau_{i+1} < \tau_i$, so $T_{\tau_{i+1}} \supseteq T_{\tau_i}$, and we need to incrementally report new results.

\newcommand{\tauprev}{\ensuremath{{\tau_\prec}}}
%
\begin{definition}[\textbf{\problemtwo}]
    Given $(P,\dist, I)$ and $\tauprev > \tau > 0$, it asks to report all $\tau$-durable triangles (or $\eps$-triangles) that are not $\tauprev$-durable, along with their lifespans.    
\end{definition}

\mparagraph{Aggregate-durable pairs}
%
Given a pair $(p_1,p_2) \in P \times P$, we consider the set $U$ of nodes incident to both $p_1$ and $p_2$ and aggregate the lifespans of triplets $(u, p_1, p_2)$.
We call $U$ the {\em witness} of $(p_1, p_2)$.
There are two natural ways of aggregating over $U$: \SUM\ and \UNION.
For \SUM, we aggregate by summing up the durabilities of triplet lifespans, i.e.,
$\textsf{AGG}(p_1,p_2,U) = \sum_{u\in U} \left|I(u, {p_1} ,{p_2})\right|$.
For \UNION, we aggregate by first taking the union of the triplet lifespans and then considering its length, i.e.,
$\textsf{AGG}(p_1,p_2,U) = \left|\bigcup_{u\in U} I(u, {p_1} ,{p_2})\right|$.
Intuitively, \SUM\ gives higher weights to time periods when multiple simultaneous connections exist,
while \UNION\ only cares about whether a period is covered at all by any connection.
Given durability parameter $\tau > 0$, a pair $(p_1, p_2)\in P \times P$ is {\em $\tau$-aggregate-durable}
if $\dist(p_1, p_2) \leq 1$ and $\textsf{AGG}(p_1,p_2,U) \ge \tau$ for $U = \{ u \in P \mid \dist(p_1,u), \dist(p_2, u) \le 1 \}$.
We also define \emph{$\tau$-aggregate-durable $\eps$-pairs}
by relaxing the distance thresholds for $\dist(p_1, p_2)$, $\dist(p_1,u)$, and $\dist(p_2, u)$ from $1$ to $1+\eps$.
Let $K_{\tau},K_{\tau}^{\eps}$ be the set of all $\tau$-aggregate-durable pairs, $\eps$-pairs respectively.
Notice that $K_{\tau}\subseteq K_{\tau}^{\eps}$.

\begin{definition}[\textbf{\problemthree}]
    Given $(P,\dist, I)$ and $\tau \ge 0$,
it asks to report all $\tau$-aggregate-durable pairs (or $\eps$-pairs).
\end{definition}

Suppose we have embedded the co-authorship graph in Example~\ref{exmpl:2} as a proximity graph where nodes represent authors.
The goal is to find pairs of coauthors $p_i, p_j$ who have collaborated sufficiently with various others,
either in terms of total time over all collaborators (\SUM),
or over a large portion of $p_i$ and $p_j$'s shared active lifespan (\UNION).

\begin{table}[t]
    \centering
    \begin{tabular}{c|c}
     \toprule
    Problem & Time complexity in $\O(\cdot)$\\\hline
    \problemone & $n\eps^{-O(\ddim)}+\OUT$\\\hline
    \problemtwo & $\eps^{-O(\ddim)} \cdot \OUT$\\\hline
    \problemthree--\SUM\ & $\eps^{-O(\ddim)} \cdot (n+\OUT)$  \\\hline
    \problemthree--\UNION\ & $\kappa\eps^{-O(\ddim)} \cdot (n+\OUT)$  \\
    \bottomrule
    \end{tabular}
    \vspace{1em}
    \caption{Summary of our main results.
    Here, $n$ is the input size, i.e., the number of points in $P$; $\ddim$ is the doubling dimension of $(P,\dist)$; $\tau$ is the durability parameter; $\OUT$ is the output size for the respective problem (different for each problem); $\eps$ is the approximation ratio; and $k$ is the parameter used for $(\tau,\kappa)$-\UNION\ durability. In the complexities reported above, $\O(\cdot)$ hides a $\polylog n$ factor, and the hidden constants in $O(\cdot)$ and $\O(\cdot)$ may depend on $\ddim$, which is assumed to be a constant.
    }
    \label{table:results}
\end{table}

\subsection{Our Results and Approach}
\label{sec:intro:results}

We present algorithms for $\eps$-approximate versions of all three problems, whose time complexity are summarized in Table~\ref{table:results}.
In all cases, we report all durable triangles along with some durable $\eps$-triangles.
The running time is always near-linear in terms of the input and output size, which is almost the best one could hope for. Our solutions leverage the observation that proximity graphs in practice often have bounded \emph{spread} $\Delta$ and \emph{doubling dimension} $\ddim$ ---
Section~\ref{sec:preliminary} further reviews these concepts and the associated assumptions.
The complexities in Table~\ref{table:results} assume spread to be $n^{O(1)}$ (hence $\Delta$ is omitted) and doubling dimension to be constant,
but our algorithms also work for more general cases.

For \textbf{\problemone\ (Section~\ref{sec:offline})}, our main approach is to construct a hierarchical space decomposition consisting of a canonical set of balls, via a \emph{cover tree}.
We use the canonical set of balls to obtain a compact, implicit representation of points within unit distance from each point, and then use an {\em interval tree} along with auxiliary data structures to report durable triangles in linear time.
The algorithm runs in $\O(n\eps^{-O(\ddim)}+\OUT)$ time, where $\OUT \in [|T_{\tau}|,|T_{\tau}^\eps|]$ is the result size. (The $\O(\cdot)$ notation hides polylogarithmic factors).
For the $\ell_\infty$ metric, this result can be improved to $\O(n+|T_\tau|)$.
Moreover, our data structures can be extended to support delay-guaranteed enumeration as well as dynamic settings where nodes are inserted or deleted according to their lifespan.

For \textbf{\problemtwo\ (Section~\ref{sec:online})}, to support incremental computation of queries arriving in an online setting, we additionally maintain an activation threshold for each point with respect to different durability parameters.
In more detail, for each durability parameter $\tau$, we design an oracle that can efficiently find the largest value $\beta < \tau$ such that $p$ participates in a $\beta$-durable triangle that is not $\tau$-durable, which is key to achieve near-linear time complexity.
Our algorithm constructs an $\O(n)$-size data structure in $\O(n\eps^{-O(\ddim)})$ time,
such that given the previous query parameter $\tauprev$ and current query parameter $\tau < \tauprev$,
it can report the delta results in $\O(\eps^{-O(\ddim)}\cdot \OUT)$ time, where $\OUT \in \left[|T_{\tau} - T_{\tauprev}|, |T_{\tau}^\eps -T_{\tauprev}^\eps|\right]$ is the delta result size.
Specifically, for the $\ell_\infty$ metric, \problemtwo\ can be solved exactly in $\O(|T_{\tau_{i+1}}\setminus T_{\tau_i}|)$ time.

\newchange{
For \textbf{\problemthree\ (Section~\ref{sec:aggregate})}, recall that the problem requires aggregating lifespans over witness set $U$.
We build auxiliary data structures to compute the sum or union of intervals intersecting any given interval.
Additionally, we identify a special ordering of $P$ such that only a bounded number of pairs that are not aggregate-durable will be visited, so the linear-time complexity can be guaranteed.
For the \SUM\ version of the problem, we present an $\O((n+\OUT)\cdot\eps^{-O(\ddim)})$-time algorithm, where $\OUT \in \left[K_\tau, K^\eps_\tau\right]$ is the output size.
The \UNION\ version is more challenging because of the inherent hardness of computing the union of intervals that intersect a query interval.
However, as shown in Section~\ref{sec:aggregate:union}, we can still get an near-linear-time and output-sensitive algorithm that reports all $\tau$-\UNION-durable pairs along with some $(1-1/e)\tau$-\UNION-durable $\eps$-pairs.
}


\paragraph{Extensions.}
%
Our algorithms also work for every $\ell_\alpha$-metric%
\footnote{If $\dist$ is the \emph{$\ell_\alpha$-metric} then $\dist(p,q)=\left(\sum_{j=1}^d|p_j-q_j|^\alpha\right)^{1/\alpha}$,
where $p_j, q_j$ are the $j$-th coordinates of points $p$ and $q$, respectively.}
or metric with bounded expansion constant%
\footnote{A metric space $(P,\dist)$ has \emph{expansion constant} $D$
if $D$ is the smallest value such that for every $p\in P$ and $r>0$ $|P\cap \mathcal{B}(p,2r)|\leq D\cdot|P\cap \mathcal{B}(p,r)|$,
where $\mathcal{B}(p,r)$ is the ball with center $p$ and radius $r$.}.
Moreover, all our results for reporting triangles can be extended to reporting cliques, paths, and star patterns of constant size.
See details in Appendix~\ref{appndx:extensions}.

\newchange{
\mparagraph{Connection with triangle listing algorithms in general graphs} Consider simple directed or undirected graphs with $n$ vertices and $m$ edges. The trivial algorithm by listing all triples of vertices runs in $O(n^3)$ time. This is worst-case optimal in terms of $n$, since a dense graph may contain $\Theta(n^3)$ triangles.
A graph with $m$ edges contains $O(m^{3/2})$ triangles.
It has been shown that all triangles in a graph of $m$ edges can be enumerated in $\O(m^{3/2})$ time~\cite{itai1977finding, ngo2018worst, veldhuizen2014leapfrog}. This is also worst-case optimal, since a graph of $m$ edges may contain $\Theta(m^{3/2})$ triangles. Later, output-sensitive algorithms for listing triangles were developed using fast matrix multiplication, which run in $\O(n^\omega+ n^{\frac{3(\omega-1)}{5-\omega}} \cdot \OUT^{\frac{2(3-\omega)}{5-\omega}})$ or $\O(m^{\frac{2\omega}{\omega+1}}+ m^{\frac{3(\omega-1)}{\omega+1}}\cdot \OUT^{\frac{3-\omega}{\omega+1}})$ time, where $ O(n^\omega)$ is the running time of $n\times n$ matrix multiplication and $\OUT$ is the number of triangles in the graph~\cite{bjorklund2014listing}. In contrast, it has been shown~\cite{patrascu2010towards} that listing $m$ triangles in a graph of $m$ edges requires $\Omega(m^{4/3-o(1)})$ time, assuming the 3SUM conjecture\footnote{The 3SUM conjecture states that: Given three sets $A,B,C$ of $n$ elements, any algorithm requires $\Omega(n^{2-o(1)})$ time to determine whether there exists a triple $(a,b,c) \in A \times B \times C$ such that $a + b+c =0$.}. A careful inspection of this lower bound construction reveals that listing $n^{3/2}$ triangles in a graph of $n$ vertices requires $\Omega(n^{2-o(1)})$ time, assuming the 3SUM conjecture. These lower bounds together rule out the possibility of listing triangles in general graphs within $O(m+n+\OUT)$ time, unless the 3SUM conjecture is refuted.

Existing techniques for listing triangles in general graphs do not yield efficient algorithms for our setting and several new ideas are needed to obtain the results of this paper. First, most traditional techniques for listing triangles do not handle temporal constraints on vertices or edges. Recently, efficient algorithms for durable--join with temporal constraints on edges are proposed in~\cite{hu2022computing}. However, their algorithm requires $\Omega(m^{3/2})$ time for listing durable triangles, even if the number of durable triangles is much smaller. Even without temporal constraints, all the worst-case optimal join algorithms run in super-linear time, in terms of $n$ and $\OUT$, for listing triangles in proximity graphs. Furthermore, in our setting, the input is an implicit representation of a proximity graph $(P,\dist)$. To feed $(P,\dist)$ as input to these algorithm, the number of edges $m$ can be quadratic in terms of $|P|$, which already requires $\Omega(n^2)$ time for processing the input, not to mention the time for identifying triangles.
Our algorithms use novel geometric data structures to identify all triangles that a point belongs to in time which is linear (ignoring $\log n$ factors) to both input size $n$ and output size $\OUT$.
Finally, the known algorithms do not handle the incremental or the aggregate versions of our problem, and we need completely new techniques to further exploit the structure of proximity graphs. 
 
}
\section{Preliminaries}
\label{sec:preliminary}

We start by reviewing some basic concepts and data structures.
Building on the basic data structures, we introduce an oracle that will be frequently used by our algorithms in the ensuing sections.

\subsection{Basic concepts and data structures}
\label{subsec:structures}

\noindent {\bf Spread.} The \emph{spread} of a set $P$ under distance metric $\dist$ is the ratio of the maximum and minimum pairwise distance in $P$.
For many data sets that arise in practice, the spread is polynomially bounded in $n$, and this assumption is commonly made in machine learning and data analysis~\cite{beygelzimer2006cover, cunningham2021k, kumar2008good, borassi2019better}.

\paragraph{Doubling dimension.}
For $x\in \Re^d$ and $r\geq 0$,
let $\mathcal{B}(x,r)=\{y\in \Re^d\mid \dist(x,y)\leq r\}$ denote the ball (under the metric $\dist$) centered at point $x$ with radius $r$.
A metric space $(P,\dist)$ has \emph{doubling dimension} $\ddim$ if
for every $p\in P$ and $r>0$, $\mathcal{B}(p,r) \cap P$ can be covered by the union of at most $2^\ddim$ balls of radius $r/2$.
For every $\alpha>0$, let $\ell_\alpha$ be the $\alpha$ norm.
The metric space $(P,\ell_\alpha)$ has doubling dimension $d$ for every $P\subset\Re^d$,
but for specific $P\subset\Re^d$, the doubling dimension can be much smaller---e.g.,
points in $3$d lying on a $2$-dimensional plane or sphere has doubling dimension $2$.
Doubling dimensions and their variants are popular approaches for measuring the intrinsic dimension of a data set in high dimension;
see, e.g.,~\cite{DBLP:conf/iclr/PopeZAGG21, facco2017estimating, bshouty2009using, DBLP:conf/iclr/PopeZAGG21}.
It has been widely shown that graphs arising in practice
have low doubling dimension~\cite{ng2002predicting, tenenbaum2000global, feldmann2020parameterized, damian2006distributed, zhoudoubling, verbeek2014metric}.
\newchange{Empirical studies in these papers and other sources (e.g., \cite{linkDdim}) show that the doubling dimension of router graphs, internet latency graphs, citation graphs, and movie database graphs are less than $15$.}

\newchange{
\paragraph{Interval tree.}
Let $\mathcal{I}$ be a set of intervals.
An interval tree~\cite{mark2008computational} is a tree-based data structure that can find intersections of a query interval $I$ with the set of intervals $\mathcal{I}$ stored in the interval tree.
For example, it can report or count the number of intervals in $\mathcal{I}$ intersected by $I$ visiting only $O(\log n)$ nodes. It has $O(n)$ space and it can be constructed in $O(n\log n)$ time.}

\paragraph{Cover tree.}
A cover tree $\mathcal{T}$ is a tree-based data structure where
each node $u$ of $\mathcal{T}$ is associated with a representative point $\rep_u \in P$ and a ball $\mathcal{B}_u$.
Each node belongs to an integer-numbered level; if a node $u$ is at level $i$ then its children are at level $i-1$.
Let $C_i$ be the set of balls associated with nodes at level $i$.
The radius of each ball $\mathcal{B}_u$ at level $i$ is $2^i$ (notice that our definition allows level numbers to be positive or negative).
Each point $p\in P$ is stored in one of the leaf nodes.
The root consists of a ball that covers the entire data set \change{and its representative point is any point in $P$}.
A cover tree satisfies the following constraints:
\begin{itemize}[leftmargin=*]
    \item {\bf ({Nesting})} If there is a node $u$ at level $i$ with a representative point $\rep_u\in P$,
    then $\rep_u$ is also a representative point in a node at level $i-1$.
    \item {\bf ({Covering})} For every representative point $\rep_u$ at level $i-1$,
    there exists at least one representative $\rep_v$ at level $i$ such that $\dist(\rep_v,\rep_u)<2^i$.
    We designate $v$ as the parent of $u$.
    \item {\bf ({Separation})} For every $u, v$ at level $i$, $\dist(\rep_u,\rep_v)>2^i$.
\end{itemize}
Traditionally, a cover tree is used mostly for approximate nearest-neighbor queries~\cite{beygelzimer2006cover, har2005fast}.
\newchange{We modify the construction of the cover tree to use it for \emph{ball-reporting queries} in bounded doubling spaces.
Given a point $p$, let $\mathcal{B}(p,r)=\{x\in\Re^d\mid \dist(x,p)\leq r\}$ be a ball with radius $r$ centered at $p$, and let $\mathcal{B}(p):=\mathcal{B}(p,1)$. Given a query point $p$, the goal is to report $B(p)\cap P$ efficiently. We modify the cover tree to answer ball-reporting queries approximately.
}
\newchange{In each node $u$ of the cover tree, we (implicitly) store $P_u$, i.e., the points that lie in the leaf nodes of the subtree rooted at $u$.
Let $p$ be a query point. We find a set of nodes in the cover tree whose associated balls entirely cover $\mathcal{B}(p)$ and might cover some region outside $\mathcal{B}(p)$ within distance $(1+\eps)$ from the center of $\mathcal{B}$. The set of nodes we find in the query procedure are called \emph{canonical nodes}, their corresponding balls are called \emph{canonical balls}, and the subsets of points stored in the canonical balls are called \emph{canonical subsets}. In the end, we report all points stored in the canonical nodes.}
\newchange{More formally, in Appendix~\ref{appndx:covertree}, we show how to construct a data structure with space $O(n)$ in $O(n\log n)$ time,
while achieving the following guarantees when the spread is bounded.
 For a query point $q\in\Re^d$, in $O(\log n + \eps^{-O(\ddim)})$ time, 
    it returns a set of $O(\eps^{-O(\ddim)})$ canonical balls (corresponding to nodes in the modified cover tree) of diameter no more than $\eps$,
    possibly intersecting,
    such that each point of $\mathcal{B}(q)\cap P$ belongs to a unique canonical ball.
    Each canonical ball may contain some points of $\mathcal{B}(q, 1+\eps) \cap P$.
}
%


\subsection{Durable ball query}
\label{sec:tau-durable-ball}
\newchange{In this subscection, we describe an extension of the ball-reporting query that will be frequently used by our algorithms.
Given $(P,\dist, I)$, $\tau > 0$, and a point $p$ with interval $I_p$,
a \emph{$\tau$-durable ball query} finds all points $q \in P$ such that $\dist(p,q) \le 1$, $|I_p \cap I_q| \ge \tau$, and $I_p^- \in I_q$.
Answering such a query exactly is inherently expensive even in the Euclidean space, since a near-linear space data structure has $\Omega(n^{1-1/d} + \OUT)$ query time~\cite{chan2012optimal}, where $\OUT$ is the output size.
If we use such a data structure for our problem in metrics with bounded doubling dimension, it would lead to a near-quadratic time algorithm for the \problemone problem. Instead, we consider the following relaxed version:} 

\begin{definition}[$\eps$-approximate $\tau$-durable ball query]
    \label{def:def1}
    Given $(P,\dist, I)$, $\tau \ge 0$, \change{$\eps\in(0,1)$}, and a point $p$ with interval $I_p$,
    find a subset $Q \subseteq P$ of points such that
    $\mathcal{B}(p) \cap P \subseteq Q \subseteq \mathcal{B}(p, 1+\eps) \cap P$,
    and for every $q \in Q$, $|I_p \cap I_q| \ge \tau$ and $I_p^- \in I_q$.
\end{definition}
\newchange{We note that the condition $I_p^- \in I_q$ is needed to avoid reporting duplicate results, as we will see in the next sections.}

\paragraph{Data structure.} \newchange{Intuitively, we use a multi-level data structure $\DTC$ to handle this query.
At the first level, we construct a cover tree $\CT$ on $P$ to find a small number of canonical nodes that contain the points of $\mathcal{B}(p) \cap P$ (but may also contain some point of $\mathcal{B}(p, 1+\eps) \cap P$).
At each node $u$ of the cover tree, we construct an interval tree $\IT_u$ over the temporal intervals of $P_u$.
Using the cover tree along with interval trees, we can find a set of $O(\eps^{-d})$ canonical nodes that contain all points $q$ within distance $1$ from $p$ and $I_p^-\in I_q$, but they may also contain points within distance $1+\eps$ from $p$.
$\DTC$ uses $O(n\log n)$ space and can be constructed in $O(n\log^2{n})$ time.
}

Let $\DTCQ(p, \tau, \eps)$ denote the query procedure of $\DTC$ with parameters $p, \tau, \eps$.
It answers the query as follows:
\begin{itemize}[leftmargin=*]
    \item {\bf (step 1)} We query $\CT$ with point $p$ and radius $1$,
    and obtain a set of canonical nodes $\mathcal{C}=\{u_1, u_2, \ldots, u_k\}$ for $k = O(\eps^{-O(\ddim)})$.
    From Appendix~\ref{appndx:covertree}, each node in $\mathcal{C}$ corresponds to a ball with diameter no more than $\eps$.
    For each $u_j\in \mathcal{C}$,  $\dist(p,\rep_j)\le 1+\eps/2$.\footnote{\newchange{For simplicity, we denote the representative point of a node $u_i$ by $\rep_i$.}}
    \item {\bf (step 2)} For each canonical node $u_j \in \C$, we query $\IT_{u_j}$ with $I^-_p$,
    and obtain all $q \in \IT_u$ such that $I_q^- +\tau\leq I_p^- +\tau\leq I_q^+$.
\end{itemize}
In the end, $\DTCQ$ returns $O(k)$ disjoint result point sets, whose union is the answer to the $\eps$-approximate $\tau$-durable ball query.
The grouping of result points into subsets and the implicit representation of these subsets is an important feature of $\DTCQ$
that we shall exploit in later sections.
Note that $\DTCQ$ might return a point $q$ such that $\dist(p,q) > 1$, but $\dist(p,q) \leq 1+\eps$ always holds.
Together, we obtain:

\begin{lemma}
\label{lem:helper-2}
    Given a set $P$ of $n$ points, a data structure can be built in $O(n\log^2 n)$ time with $O(n\log n)$ space,
that supports an $\eps$-approximate $\tau$-durable ball query, computing a family of $O(\eps^{-d})$ canonical subsets in $O(\eps^{-d}\log n)$ time.
\end{lemma}


\smallskip \noindent {\bf Extended data structure with refined result partitioning.}
\newchange{
We define the more involved query procedure $\DTCEQ(p,\tau,\tau',\eps)$, which will be used by our algorithms for \problemtwo\ in Section~\ref{sec:online}. 
The goal is to return a subset $Q \subseteq P$ such that
$\mathcal{B}(q) \cap P \subseteq Q \subseteq \mathcal{B}(q, 1+\eps) \cap P$, and for every $q\in Q$, $I_q^-+\tau\leq I_p^-+\tau\leq I_q^+$ (just as for $\DTCQ$), with the additional constraint that $I_q^+\geq I_p^-+\tau'$. Let $\DTCE$ be the extended version of $\DTC$ to answer $\DTCEQ$. It consists of a cover tree along with two levels of interval trees, one to handle the first linear constraint, and the second to handle the additional linear constraint. 
The space, construction time and query time of $\DTCE$ are increased only by a $\log n$ factor compared with $\DTC$.}

\section{Reporting Durable Triangles}
\label{sec:offline}
This section describes our near-linear time algorithm for the $\eps$-approximate \problemone\ problem.
As mentioned, our algorithm works for every general metric with constant doubling dimension.
In Appendix~\ref{appndx:linfty}, we show how to solve the problem exactly for $\ell_\infty$ metric.

\mparagraph{High-level Idea}
We visit each point $p\in P$ with $|I_p| \ge \tau$ (a prerequisite for $p$ to be in a $\tau$-durable triangle), and report all $\tau$-durable triangles that $p$ anchors. 
To find all $\tau$-durable triangles anchored by $p$,
we run a $\tau$-durable ball query around $p$ on $\DTC$ (Section~\ref{sec:tau-durable-ball})
to get an implicit representation (as a bounded number of canonical subsets) of all points within distance $1$ from $p$,
where $I_p^-$ is the largest left endpoint among their lifespans \newchange{($p$ should be the newest point among the three points of each triangle we report)}.
Recall that each canonical subset returned consists of points within a ball of a small diameter,
so we can approximate inter-ball distances among points by the distances among the ball centers.
For every pair of balls, if their centers are within distance $1$ (plus some slack),
we report all $\tau$-durable triangles consisting of $p$ and the Cartesian product of points in the two balls.

\begin{algorithm}[t]
\caption{{\sc ReportTriangle}$(\DTC, p, \tau, \eps)$}
\label{alg:offline-2}

$\C_{p}: \{\C_{p,1}, \C_{p,2}, \cdots, \C_{p,k}\} \gets \DTCQ(p,\tau,\eps/2)$,
    with $\rep_i$ denoting the representative point of the ball for $\C_{p,i}$\;
\ForEach{$j \in [k]$}{
    \ForEach{$q, s \in \C_{p,j}$ where $q$ precedes $s$}{
        \KwSty{report} $(p,q,s)$\;\label{alg:offline-2:same-ball}
    }
}
\ForEach{$i,j \in [k]$ where $i < j$}{
    \If{$\dist(\rep_i, \rep_j) \le 1 + \frac{\eps}{2}$}{
        \ForEach{$(q,s) \in \C_{p,i} \times \C_{p,j}$}{
            \KwSty{report} $(p,q,s)$\;\label{alg:offline-2:diff-ball}
        }
    }
}
\end{algorithm}
\paragraph{Algorithm.}
As a preprocessing step, we construct the data structure $\DTC$ as described in Section~\ref{sec:tau-durable-ball} over $P$.
Our algorithm invokes \textsc{ReportTriangle} (Algorithm~\ref{alg:offline-2}) for each point $p \in P$.
\textsc{ReportTriangle} runs a $\tau$-durable ball query $\DTCQ(p,\tau,\eps/2)$%
---note the use of $\eps/2$ here for technical reasons---%
and obtains a family of disjoint result point sets $\mathcal{C}_p=\{\C_{p,1}, \C_{p,2}, \ldots, \C_{p,k}\}$ for some $k = O(\eps^{-\ddim})$
(see also Figure~\ref{fig:offline-2}).
Each $\C_{p,j}$ is covered by a cover tree ball in $\DTC$ with diameter of no more than $\eps/2$,
and contains all points therein whose intervals ``sufficiently intersect'' $I_p$,
i.e., any $q$ in the ball satisfying $I_q^-+\tau \le I_p^-+\tau \le I_q^+$,
as explained in Section~\ref{sec:tau-durable-ball}.
\begin{figure}
    \centering
    \includegraphics[scale=0.19]{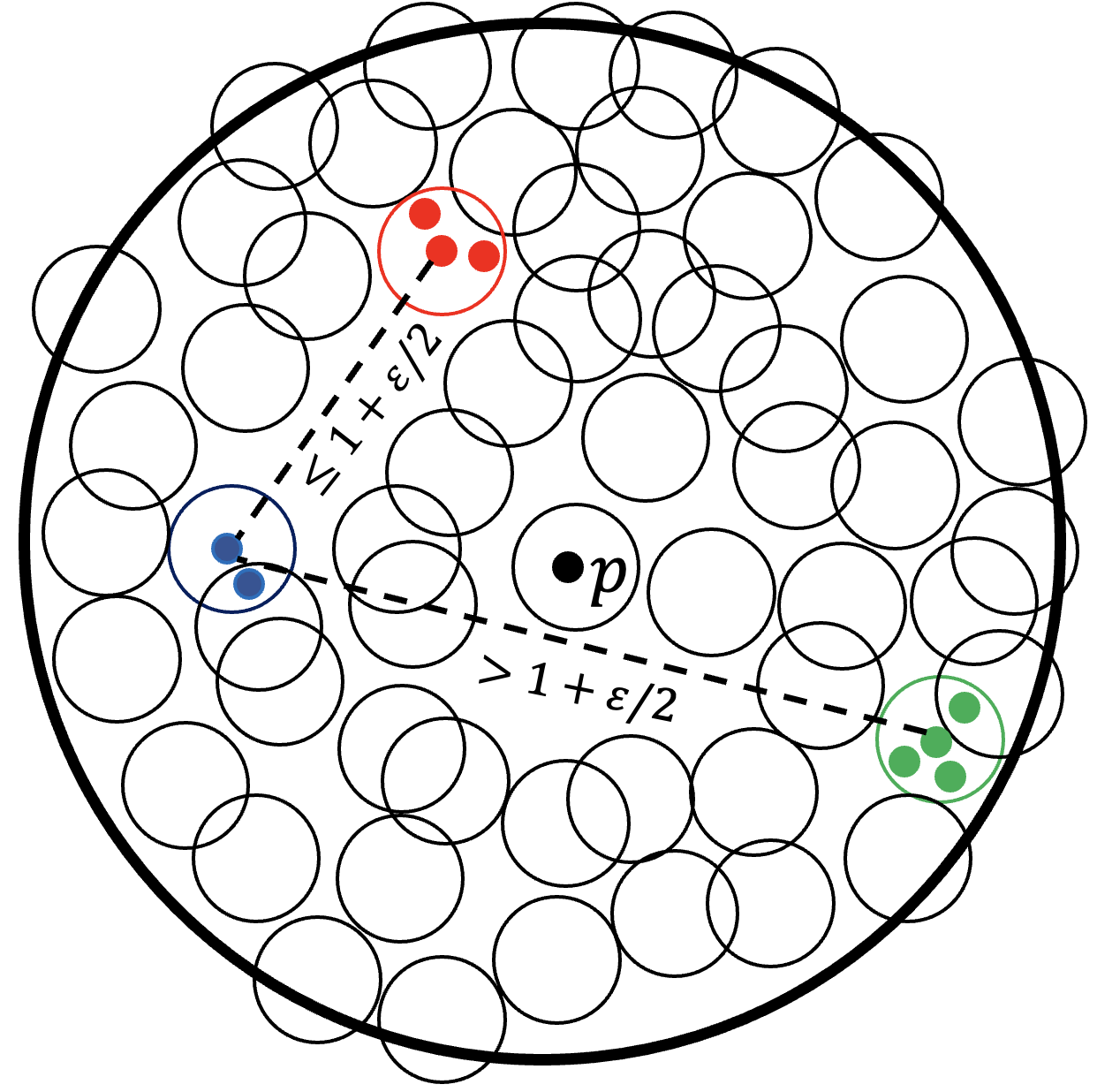}
    \caption{Illustration of Algorithm~\ref{alg:offline-2}: $p$ is visited.
    The small (possibly overlapping) balls represent the canonical nodes returned from $\DTC$.
    Each point within distance $1$ from $p$ lies in exactly one such ball.
    We report the triangles formed by $p$ and the points in red and blue balls that satisfy the durability constraint.
    We do not report triangles formed by $p$ and the points in blue and green balls because they are well separated.}
    \label{fig:offline-2}
\end{figure}

Given $p$, all $\tau$-durable triangles $(p,q,s)$ anchored by $p$ can be classified into two types:
(1)~$q$ and $s$ belong to the same result point set $\C_{p,j}$ for some $j$; and
(2)~$q$ and $s$ belong to different sets $\C_{p,i}$ and $\C_{p,j}$ (where $i \ne j$) that are sufficiently close.
To report triangles of the first type, we simply enumerate all pairs of $q$ and $s$ within $\C_{p,j}$, for each $j$.
We avoid duplicate reporting of $(p,q,s)$ and $(p,s,q)$ by always picking $q$ as the point with the smaller index in $\C_{p,j}$.
To report triangles of the second type, we consider $(i,j)$ pairs where $\rep_i$ and $\rep_j$, the representative points of the balls containing $\C_{p,i}$ and $\C_{p,j}$, are within distance $1+\eps/2$.
We simply enumerate the Cartesian product of $\C_{p,i}$ and $\C_{p,j}$.
We avoid duplicate reporting of $(p,q,s)$ and $(p,s,q)$ by imposing the order $i < j$.

\newchange{
\mparagraph{Correctness}
Let $(p,q,s)$ be a triangle reported by our algorithm. We show that $(p,q,s)$ is a $\tau$-durable $\eps$-triangle. From Section~\ref{sec:tau-durable-ball}, we know that $\dist(p,q)\leq 1+\eps/2$ and $\dist(p,s)\leq 1+\eps/2$, because $q$ and $s$ belong in one or two canonical subsets in $\C_p$. If $q$ and $s$ belong to the same canonical subset $\C_{p,i}$ then by definition $\dist(q,s)\leq \eps/2$. If $q\in \C_{p,i}$ and $s\in \C_{p,j}$ for $i\neq j$, then $\dist(q,s)\leq \dist(q,\rep_i)+\dist(s,\rep_j)+\dist(\rep_i,\rep_j)\leq \eps/4+\eps/4+(1+\eps/2)\leq 1+\eps$. In every case, it is true that $\dist(p,q), \dist(p,s), \dist(s,q)\leq 1+\eps$. Hence, $(p,q,s)$ is an $\eps$-triangle. Next, we show that $|I_p\cap I_q\cap I_s|\geq \tau$. Recall that by definition, $|I_p|\geq \tau$. Using the $\tau$-durable ball query $\DTCQ(p,\tau,\eps/2)$, we have that $|I_p\cap I_q|\geq \tau$, $|I_p\cap I_s|\geq \tau$, $I_p^-\in I_q$, and $I_p^-\in I_s$ (see also Definition~\ref{def:def1}). We can rewrite these inequalities as $I_q^- +\tau\leq I_p^-+\tau\leq I_q^+$ and $I_s^- +\tau\leq I_p^-+\tau\leq I_s^+$. Hence, $I_p^-+\tau\leq \min\{I_q^+, I_s^+\}$, concluding that $|I_p\cap I_q\cap I_s|\geq \tau$. Each triangle $(p,q,s)$ is reported only once, in a specific vertex order:
the temporal conditions ensure that $p$ anchors the triangle,
and the ordering of $q$ and $s$ is consistently enforced by \textsc{ReportTriangle}.
Overall, we showed that if $(p,q,s)$ is reported, then it is a $\tau$-durable $\eps$-triangle and it is reported exactly once.

Next, we prove that we do not miss any $\tau$-durable triangle. Let $(p,q,s)$ be a $\tau$-durable triangle. Without loss of generality, assume that $I_p^-\geq \max\{I_q^-,I_s^-\}$. By definition, $\dist(p,q)\leq 1$, $\dist(p,s)\leq 1$, and $\dist(q,s)\leq 1$. Hence, after visiting $p$, by the definition of the $\tau$-durable ball query $\DTCQ(p,\tau,\eps/2)$, there exist indexes $i, j$ such that $q\in \C_{p,i}$ and $s\in \C_{p,j}$.
If $i=j$, then $(p,q,s)$ must be reported by Line~\ref{alg:offline-2:same-ball} of \textsc{ReportTriangle}.
If $i\neq j$, note that $\dist(\rep_i,\rep_j) \leq \dist(\rep_i,q)+\dist(q,s)+\dist(s,\rep_j) \leq \eps/4+1+\eps/4\leq 1+\eps/2$;
therefore $(p,q,s)$ must be reported by Line~\ref{alg:offline-2:diff-ball}. Overall, we showed that every $\tau$-durable triangle is reported by Algorithm~\ref{alg:offline-2}.

}

\mparagraph{Time complexity}
By Lemma~\ref{lem:helper-2}, we can construct $\DTC$ in time $O(n\log^2 n)$.
For each $p \in P$, we run a $\tau$-durable ball query on $\DTC$ in $O(\eps^{-O(\ddim)}\log n)$ time.
Moreover, $|\C_p|=O(\eps^{-O(\ddim)})$.
Checking pairs in $\mathcal{C}$ in which to search for triangles of Type~(2) takes additional $O(\eps^{-2\cdot O(\ddim)})$ time.
All pairs of points examined by \textsc{ReportTriangle} are indeed returned,
and together they correspond to all $\tau$-durable triangles plus some $\eps$-triangles involving $p$.
Hence, the overall additional time incurred is $O(\OUT)$ where $|T_{\tau}|\leq \OUT \leq |T_{\tau}^\eps|$.

\begin{theorem}
    \label{thm:offline-gen}
    Given $(P,\phi,I)$, $\tau > 0$, and $\eps > 0$,
    $\eps$-approximate \problemone\ can be solved in
    $O\left(n(\eps^{-O(\ddim)}\log n + \log^2 n) + \OUT\right)$ time,
    where $n = |P|$, $\ddim$ is the doubling dimension of $P$, and $\OUT$ is the number of triangles reported.
\end{theorem}

\smallskip \noindent {\bf Remark 1.}
For every $\ell_\alpha$ norm in $\Re^d$, we can simplify the data structure $\DTC$ using a quadtree instead of a cover tree.
The running time and approximation with respect to the overall number of reported triangles remain the same.

\smallskip \noindent {\bf Remark 2.}
Our algorithm can be extended to support \emph{delay-guaran-}{\em teed enumeration}~\cite{berkholz2017answering, idris2017dynamic, agarwal2021dynamic} of durable patterns,
i.e., the time between reporting two consecutive patterns is bounded.
After spending $O(n(\eps^{-O(\ddim)}\log n + \log^2 n))$ preprocessing time,
we can support $O(\eps^{-O(\ddim)}\log n)$-delay enumeration for $\eps$-approximate \problemone.

\smallskip \noindent {\bf Remark 3.}
Using a dynamic cover tree, we can extend our algorithm to the dynamic setting where we do not have all points upfront.
If points are inserted or deleted according to their lifespans, we support $O(\log^3 n)$ amortized update time.
After inserting a point $p$, we can report the new (if any) triangles that $p$ participates in using Algorithm~\ref{alg:offline-2})
in time near linear to the number of new triangles reported.
We show the details in Appendix~\ref{appndx:dynamic}.
\section{Incremental Reporting When Varying $\pmb\tau$}
\label{sec:online}

We next consider reporting durable triangles when queries with different durability parameters arrive in an online fashion.
As discussed in Section~\ref{sec:intro:defs}, the problem, \problemtwo, boils down to reporting any new result triangles in $T_\tau \setminus T_\tauprev$,
where $\tauprev > \tau$ are the previous and current durability parameters, respectively.

As a starter, we can proceed similarly as in Section~\ref{sec:offline},
reporting durable triangles for each anchor point $p$,
but taking care to ensure that we report only $\tau$-durable triangles that are not $\tauprev$-durable.
Doing so entails retrieving candidate pairs $(q, s)$ as in \textsc{ReportTriangle},
but additionally guaranteeing that \emph{at least one} of $q$ and $s$ ends between $I_p^- + \tau$ and $I_p^- + \tauprev$,
which leads to $I(p,q,s)$ having durability between $\tau$ and $\tauprev$.
This additional search condition necessitates the modified data structure $\DTCE$ discussed in Section~\ref{sec:preliminary}.

However, the naive approach above has the following problem.
It is possible that we carry out the search on $\DTCE$ for $p$, only to realize that in the end no new result triangle needs to be reported.
Ideally, we instead want an \emph{output-sensitive} algorithm whose running time depends only on the output size.
To this end, we need an efficient way to test whether $p$ should be \emph{activated} for output; i.e., there is at least one triangle in $T_\tau \setminus T_\tauprev$ anchored by $p$.
This test motivates the idea of activation thresholds below.

\begin{definition}[Activation threshold]
\label{def:delta}
Given $(P,\dist, I)$ and $\tau > 0$, the {\em activation threshold} of $p \in P$ with respect to $\tau$ is defined as:
\begin{align*}
    \beta^{\tau}_p = \max\{ \tau' < &\tau \mid  \exists q,s \in P: I_q^- \le I_p^-,\;I_s^- \le I_p^-,\;\text{and}
    \text{$(p,q,s)$ is $\tau'$-durable but not $\tau$-durable}\}.
\end{align*}
We set $\beta^{\tau}_p = -\infty$ if no such $\tau'$ exists.
We call $\beta^{+\infty}_p$ the \emph{maximum activation threshold} of $p$.
\end{definition}

With activation thresholds, we can easily determine whether to activate $p$:
the condition is precisely $\beta^\tauprev_p \ge \tau$.
If $\beta^\tauprev_p < \tau$, by definition of $\beta^\tauprev$,
any $\tau$-durable triangle anchored by $p$ is already $\tauprev$-durable and hence does not need to be reported;
otherwise, we need to at least report $\beta^\tauprev$-durable triangles anchored by $p$.

In the following subsections,
we first describe the algorithm for processing each activated point (Section~\ref{sec:online:report}),
and then address the problem of computing activation thresholds efficiently (Section~\ref{sec:online:thresholds}),
which requires maintaining additional data structures across queries to help future queries.
Finally, we summarize our solution and discuss its complexity (Section~\ref{sec:online:complexity}).
Appendix~\ref{sec:linfty-reporting} describes our specialized solution for the $\ell_\infty$-metric.

\subsection{Reporting for each activated point}
\label{sec:online:report}

\begin{algorithm}[t]
\caption{{\sc ReportDeltaTriangle}$(\DTCE, p, \tau, \tauprev, \eps)$}
\label{alg:report-delta-general}
$\C_{p}: \{\C_{p,1}, \C_{p,2}, \cdots, \C_{p,k}\} \gets \DTCEQ(p,\tau,\tauprev,\eps/2)$,
    with $\rep_i$ as the representative point of the ball for $\C_{p,i}$
    and $\C_{p,i} = \NewerC_{p,i} \cup \OlderC_{p,i}$\;
\ForEach{$j \in [k]$}{
    \ForEach{$q, s \in \NewerC_{p,j}$ where $q$ precedes $s$}{
        \KwSty{report} $(p,q,s)$\;
    }
    \lForEach{$(q, s) \in \NewerC_{p,j} \times \OlderC_{p,j}$}{
        \KwSty{report} $(p,q,s)$
    }
}
\ForEach{$i,j \in [k]$ where $i < j$}{
    \If{$\dist(\rep_i, \rep_j) \le 1 + \frac{\eps}{2}$}{
        \lForEach{$(q,s) \in \NewerC_{p,i} \times \NewerC_{p,j}$}{
            \KwSty{report} $(p,q,s)$
        }
        \lForEach{$(q,s) \in \NewerC_{p,i} \times \OlderC_{p,j}$}{
            \KwSty{report} $(p,q,s)$
        }
        \lForEach{$(q,s) \in \OlderC_{p,i} \times \NewerC_{p,j}$}{
            \KwSty{report} $(p,q,s)$\label{alg:report-delta-general:OldNew}
        }
    }
}
\end{algorithm}

Given an activated point $p \in P$, for which we have already determined that $\beta^\tauprev_p \ge \tau$,
we report all $\tau$-durable triangles anchored by $p$ that are not $\tauprev$-durable,
using \textsc{ReportDeltaTriangle} (Algorithm~\ref{alg:report-delta-general}) explained further below.

As discussed at the beginning of this section,
reporting triangles $(p,q,s)$ that are $\tau$-durable but not $\tauprev$-durable entails
ensuring that at least one of $q$ and $s$ ends between $I_p^- + \tau$ and $I_p^- + \tauprev$.
To this end, we use the modified data structure $\DTCE$ discussed in Section~\ref{sec:preliminary}.
We query $\DTCE$ using $\DTCEQ(p,\tau,\tauprev,\eps/2)$
to get $k=O(\eps^{-O(\ddim)})$ canonical balls of the cover tree in $\DTCE$
with representative points $\rep_1, \rep_2, \ldots, \rep_k$;
$\DTCEQ$ further partitions the result point set $\C_{p,j}$ associated with each ball centered at $\rep_j$ into two subsets
$$\NewerC_{p,j} = \smash{\left\{ q \in \C_{p,j} \mid I_q^+ < I_p^- + \tauprev \right\}} \text{, }
\OlderC_{p,j} = \smash{\left\{ q \in \C_{p,j} \mid I_q^+ \ge I_p^- + \tauprev \right\}}.$$
By definition, if $q\in \NewerC_{p,j}$, then $I_q^-\leq I_p^-$ and $I_p^- +\tau\leq I_q^+<I_p^+ +\tauprev$, while if $q\in \OlderC_{p,j}$ then $I_q^-\leq I_p^-$ and $I_q^+\geq I_p^- +\tauprev$.
See Figure~\ref{fig:updateddim} for an illustration.
\begin{figure}[t]
    \centering
    \includegraphics[scale=0.8]{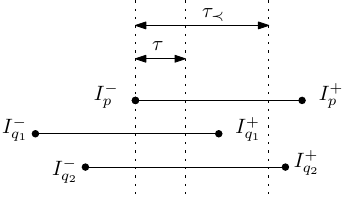}
    \caption{An illustration of $\C_{p,j} = \NewerC_{p,j} \cup \OlderC_{p,j}$.
    Here $q_1 \in \NewerC_{p,j}$ and $q_2 \in \OlderC_{p,j}$. 
    }
    \label{fig:updateddim}
\end{figure}
Recall that $\DTCEQ$ does not explicitly construct $\NewerC_{p,j}$ and $\OlderC_{p,j}$;
instead, these subsets correspond to canonical subsets of nodes in the interval trees within $\DTCE$.

For $(p,q,s)$ to be \emph{not} $\tauprev$-durable,
at least one of $q$ and $s$ must belong to some $\NewerC_{p,j}$ instead of $\OlderC_{p,j}$.
Therefore, we can divide all triangles $(p,q,s)$ that are $\tau$-durable but not $\tauprev$-durable into four types,
which can be computed with the help of the above partitioning:
(1)~$q,s \in \NewerC_{p,j}$ for some $j$;
(2)~$q \in \NewerC_{p,j}$ and $s \in \OlderC_{p,j}$ for some $j$;
(3)~$q \in \NewerC_{p_i}$ and $s \in \NewerC_{p,j}$ for some $i \neq j$ where $\rep_i$ and $\rep_j$ are sufficiently close;
(4)~$q \in \NewerC_{p_i}$ and $s \in \OlderC_{p,j}$ for some $i \neq j$ where $\rep_i$ and $\rep_j$ are sufficiently close.
\textsc{ReportDeltaTriangle} (Algorithm~\ref{alg:report-delta-general}) covers all these cases.
As with \textsc{ReportTriangle} in Section~\ref{sec:offline},
we enforce an ordering between $q$ and $s$ to ensure that only one of $(p,q,s)$ and $(p,s,q)$ is reported.
Thanks to the implicit representation of $\NewerC_{p,j}$'s and $\OlderC_{p,j}$'s,
\textsc{ReportDeltaTriangle} avoids enumerating points in a subset if they do not contribute to any result triangle.
For example, if $\OlderC_{p,i} = \emptyset$ (line~\ref{alg:report-delta-general:OldNew} of Algorithm~\ref{alg:report-delta-general}), we short-circuit the computation and avoid enumerating $\NewerC_{p,j}$.

\paragraph{Remark.}
Note that \textsc{ReportTriangle} (Section~\ref{sec:offline}) can be seen as as special case of \textsc{ReportDeltaTriangle} whenever $\tauprev>\max_{p\in P}\beta^{+\infty}_p$.

\subsection{Computing 
activation thresholds}
\label{sec:online:thresholds}

We turn to the question of how to compute the activation threshold $\beta^\tauprev$ given $p$ and $\tauprev$,
required for determining whether to activate $p$.
Naively, we can find all triangles anchored by $p$ (regardless of durability) and build a map of all activation thresholds for $p$.
However, we have a more efficient solution that builds on two ideas.
First, consider a sequence of queries with durability parameters $\tau_1, \tau_2, \ldots$.
After answering the current query, say $\tau_i$, we compute and remember $\beta^{\tau_i}_p$ for each (relevant) $p$,
so they are available to help the next query $\tau_{i+1}$.
Second, we use a binary search procedure to look for activation thresholds within a desired range,
by exploiting the extended data structure $\DTCE$
to quickly test existence of thresholds in a range without enumerating result triangles therein.
\begin{algorithm}[h]
\caption{{\sc ComputeActivation}$(\DTCE, p, \tau, \eps)$}
\label{alg:detect-general}
$I^+ \gets \{I_q^+\mid q\in P\}$, $R \gets n-1$, $L \gets 0$, $\tau_{\mathsf{ret}} \gets -\infty$\;
\While{$L\leq R$}{
$m \gets \lfloor (L+R)/2 \rfloor$\;
\lIf{$I^+[m]>I_p^-+\tau$}{$R \gets m-1$
}\lElseIf{$I^+[m]<I_p^-$}{$L \gets m+1$
}
$\tau' \gets I^+[m]-I_p^-$\;
$B \gets \textsc{DetectTriangle}(\tau', \tau)$\;
\lIf{$B=\KwSty{true}$}{$\tau_{\mathsf{ret}} \gets \tau'$, $L \gets m+1$}
\lElse{$R \gets m-1$}
}
\Return $\tau_{\mathsf{ret}}$\;

\KwSty{Subroutine} $\textsc{DetectTriangle}(\tau_1, \tau_2)$ \Begin{
$\C_{p}: \{\C_{p,1}, \C_{p,2}, \cdots, \C_{p,k}\} \gets \DTCEQ(p,\tau_1,\tau_2,\eps/2)$,
    with $\rep_i$ denoting the representative point of the ball for $\C_{p,i}$
    and $\C_{p,i} = \NewerC_{p,i} \cup \OlderC_{p,i}$\;
\ForEach{$j \in [k]$}{
    \lIf{$|\NewerC_{p,j}| \ge 2$}{
        \Return \KwSty{true}%
    }
    \lIf{$|\NewerC_{p,j}| \ge 1$ and $|\OlderC_{p,j}| \ge 1$}{
        \Return \KwSty{true}%
    }
}
\ForEach{$i,j \in [k]$ where $i < j$}{
    \If{$\dist(\rep_i, \rep_j) \le 1 + \frac{\eps}{2}$}{
        \lIf{$|\NewerC_{p,i}| \ge 1$ and $|\NewerC_{p,j}| \ge 1$}{
            \Return \KwSty{true}%
        }
        \lIf{$|\NewerC_{p,i}| \ge 1$ and $|\OlderC_{p,j}| \ge 1$}{
            \Return \KwSty{true}%
        }
        \lIf{$|\OlderC_{p,i}| \ge 1$ and $|\NewerC_{p,j}| \ge 1$}{
            \Return \KwSty{true}%
        }
    }
}
\Return \KwSty{false}\;
}
\end{algorithm}
The second idea is implemented by \textsc{ComputeActivation} (Algorithm~\ref{alg:detect-general}).
It runs a binary search making guesses for the value of $\beta^\tau_p$. For each guess $\tau'$ of $\beta^\tau_p$,
we use a primitive called \textsc{DetectTriangle} to test
whether there exists any triangle anchored by $p$ that is $\tau'$-durable but not $\tau$-durable%
---in other words, whether $\beta^\tau_p \in [\tau', \tau)$. 
\textsc{DetectTriangle} mirrors Algorithm~\ref{alg:report-delta-general},
except that it merely checks the existence of triangles for each type returning true or false instead of reporting them.
Using $\DTCEQ$ for returning implicit representations for $\NewerC_{p,j}$'s and $\OlderC_{p,j}$'s,
it is quick to check whether their combinations yield a non-empty result set.
Given $p$, the search space of activation thresholds has only $O(n)$ possibilities:
The lifespan of every triangle $(p,q,s)$ anchored by $p$ is either in $[I^-_p, I^+_q]$ or $[I^-_p, I^+_s]$, thus the durability of any triangle anchored by $p$ falls into the set $\{I_q^+-I_p^-\mid q\in P, I_q^+\geq I_p^-\}$.
The number of steps in the binary search and the number of invocations of \textsc{DetectTriangle} is $O(\log n)$.

We are now ready to put together the data structures and procedure for computing and maintaining activation thresholds.
We use two simple binary search trees \Smax\ and \Scur.
\Smax\ indexes all points $p \in P$ by their maximum activation thresholds $\beta^{+\infty}_p$.
We precompute \Smax\ by calling \textsc{ComputeActivation} for each $p \in P$.
Once constructed, \Smax\ remains unchanged across queries.

\Scur\ indexes points by their activation thresholds with respect to the durability parameter.
Suppose the current query parameter is $\tau$ and the previous one is $\tauprev$.
Before executing the current query, \Scur\ indexes each point $p$ by $\beta^\tauprev_p$,
so the current query can use \Scur\ to find $p$'s with $\beta^\tauprev_p \ge \tau$ to activate.
After completing the current query, we update \Scur\ for the next round:
as long as there exists a $\tau$-durable triangle anchored by $p$, \Scur\ indexes $p$ by the value of $\beta^\tau_p$.
Initially, \Scur\ starts out as an empty tree, which can be interpreted as having completed an initial query with durability parameter $+\infty$.

Maintenance of \Scur\ has two cases depending on the current query.
First, consider the more interesting case of $\tauprev > \tau$, where we need to potentially report new result triangles.
For each $p$ activated, i.e., $\beta^\tauprev \ge \tau$,
we call $\textsc{ComputeActivation}(\DTCE, p, \tau, \eps)$ to obtain $\beta^\tau_p$ and update $p$'s entry in \Scur.
This is all we need to do to maintain \Scur\ because,
if $p$ were not activated for the current query, we would have $\beta^\tauprev < \tau$, and therefore $\beta^\tau_p = \beta^\tauprev_p$.

In the less interesting case of $\tau \ge \tauprev$, there are no new result triangles to report, but some old ones may need to be invalidated.
Strategies for maintaining \Scur\ differ depending on the usage scenario.
In the first scenario, suppose that the client issuing the query sequence
incrementally maintains the query result as lists of triangles grouped by anchor points,
and triangles within each list are sorted by durability.
When $\tau \ge \tauprev$, the client can simply trim its lists according to $\tau$.
During this process, it can easy obtain and pass information to the server for updating \Scur:
for each anchor $p$, $\beta^\tau_p$ simply takes on the highest durability value removed from $p$'s list,
or it remains unchanged if no triangle is removed.
In the alternative (and less likely) scenario where the client does not remember anything,
the server can simply rebuild \Scur\ by running \textsc{ComputeActivation} for each $p \in \Smax$ with maximum activation threshold no less than $\tau$.

\paragraph{Correctness.}
We first show that the values $\beta_p^{\tau}$ are updated correctly in Algorithm~\ref{alg:detect-general}.
Let $\tau'$ be the parameter in the binary search that we checked in Algorithm~\ref{alg:detect-general}. 
Point $p$ can only form a $\tau'$-durable $\eps$-triangle with points $q$ whose intervals $I_q$ intersect $I_p^-$ and either $I_q^+<I_p^-+\tau$ or $I_q^+\geq I_p^-+\tau$. $\bigcup_{j}\NewerC_{p,j}$ is the set of points satisfying the first inequality, and $\bigcup_j\OlderC_{p,j}$ the set of points in the second inequality. For every pair $q,s\in \OlderC_{p,j}$, we do not activate point $p$ with durability $\tau'$. If indeed $\dist(q,s)\leq 1$ and $q,s\in \OlderC_{p,j}$, then $(p,q,s)$ is a $\tau$-durable triangle. So our algorithm does not activate a point $p$ because of a previously reported $\tau$-durable triangle.
By definition, it is also straightforward to see that $p$ should be activated at durability $\tau'$
if there is a pair of points $q, s\in \NewerC_{p,j}\cup \OlderC_{p,j}$ such as either $q$ or $s$ belongs in $\NewerC_{p,j}$. This is because either $I_q$ or $I_s$ does not overlap with $I_p$ for more than $\tau$ and overlaps more than $\tau'$, so $(p,q,s)$ was not a $\tau$-durable $\eps'$-triangle for every $\eps'>0$.
Next, let $p$ be a point that is activated because of a triangle $(p,q,s)$.
We show that any $(p,q,s)$ is a $\tau'$-durable $\eps$-triangle.
As we mentioned, either $I_q$ or $I_s$ does not intersect $I_p$ for more than $\tau$ but intersects $I_p$ for more than $\tau'$ so it remains to show that $\dist(p,q), \dist(p,s), \dist(q,s)\leq 1+\eps$.
By the definition of $\DTCE$ we have that $\dist(p,q)\leq 1+\eps/2$ and $\dist(p,s)\leq 1+\eps/2$. If $q,s$ belong in the same subset $\C_{p,j}$ then it also follows that $\dist(q,s)\leq \eps/4\leq 1+\eps$. If $q\in \C_{p,j}$ and $s\in \C_{p,i}$ for $i<j$ then in Algorithm~\ref{alg:detect-general} we only consider this triangle if and only if $\dist(\rep_j, \rep_i)\leq 1+\eps/2$. We have $\dist(q,s)\leq \dist(q,\rep_j)+\dist(\rep_j,\rep_i)+\dist(\rep_i,s)\leq 1+\eps/2+\eps/4+\eps/4=1+\eps$. So $(p,q,s)$ is a $\tau'$-durable $\eps$-triangle that is not $\tau$-durable.

The correctness of Algorithm~\ref{alg:report-delta-general} follows from
the same arguments we used to prove the correctness of Algorithm~\ref{alg:detect-general}. Overall, Algorithm~\ref{alg:report-delta-general} reports all $\tau_{i+1}$-durable triangles along with some $\tau_{i+1}$-durable $\eps$-triangles, that are not $\tau_i$-durable $\eps$-triangles. Hence,  $|T_{\tau_{i+1}} \setminus T_{\tau_i}|\leq\OUT\leq |T_{\tau_{i+1}}^{\eps} \setminus T_{\tau_i}^\eps|$.

\subsection{Solution summary and complexity}
\label{sec:online:complexity}

In summary, we build the data structure $\DTCE$ as described in Section~\ref{sec:tau-durable-ball};
its size is $O(n\log^2 n)$, and it can be constructed in $O(n\log^3 n)$ time.
We also build the index \Smax\ of maximum activation thresholds, which has size $O(n)$.
To construct \Smax, as mentioned, we perform at most $O(\log n)$ guesses for each point,
and each guess invokes \textsc{DetectTriangle} once, which takes $O(\eps^{-O(\ddim)}\log^2 n)$ time;
therefore, the total construction time for \Smax\ is $O(n\eps^{-O(\ddim)}\log^3 n)$.
Finally, we maintain the index \Scur\ of activation thresholds for the current durability parameter;
its size is $O(n)$, its initial construction time is $O(1)$, and its maintenance time will be further discussed below.

To report new result triangles when the durability parameter changes from $\tauprev$ to $\tau$,
we use \Scur\ to search for points $p$ with $\beta^\tauprev_p \ge \tau$ to activate.
Each activated point $p$ requires $O(\OUT_p+\eps^{-O(\ddim)}\log^{2}n)$ time for \textsc{ReportDeltaTriangle} to report all new durable triangles anchored by $p$,  where $\OUT_p$ denotes the number of them.
Then, to maintain \Scur, we need $O(\eps^{-O(\ddim)}\log^{3}n)$ time for \textsc{ComputeActivation},
and $O(\log n)$ time to update \Scur\ for each point $p$ activated.
For each activated $p$, at least one new durable triangle is reported, so the number of calls to  $\textsc{ComputeActivation}$ is bounded by the output size.
Overall, we spend $O\left(\OUT + \eps^{-O(\ddim)}\log^2 n \right)$ time for reporting and $O(\OUT\cdot\eps^{-O(\ddim)}\log^3 n)$ time for maintenance,
where $\OUT$ is the number of results reported.

\begin{theorem}
\label{thm:online}
  Given $(P,\dist, I)$, and $\eps > 0$, 
  a data structure of size $O(n\log^{2}n)$ can be constructed in $O(n\eps^{-O(\ddim)}\log^{3}n)$ time such that, 
  the $\eps$-approximate \problemtwo\ problem can be solved in
  $O\left(\OUT \cdot \eps^{-O(\ddim)}\log^3 n\right)$ time,
  where $n = |P|$, $\ddim$ is the doubling dimension of $P$, and $\OUT$ is the number of results reported.
\end{theorem}

\section{Reporting Aggregate-Durable Pairs}
\label{sec:aggregate}

\subsection{\SUM}
\label{sec:aggregate:sum}

We start by describing a data structure that allows us to efficiently compute
the total length of all intersections between a query interval with a given set of intervals.
Then we show how to use this primitive to report all $\tau$-\SUM-durable pairs for \problemthree-\SUM,
along with some $\tau$-\SUM-durable $\eps$-pairs (but no other pairs).

\paragraph{Interval-\SUM-durability.}
Given a set of intervals $\mathcal{I}$, we want a primitive that can efficiently decide,
given any query interval $J$ and $\tau >0$, whether $\sum_{I \in \mathcal{I}} |I \cap J| \geq \tau$.
To this end, we construct a data structure \ITSUM\ over $\I$, which is a variant of an interval tree
where each tree node $v$ is annotated with the following information: 
\begin{itemize}[leftmargin=*]
    \item $|v|$, the total number of intervals stored at $v$;
    \item $\sum_{I \in v} |I|$, the total length of intervals stored at $v$;
    \item $\sum_{I \in v} I^+$, the sum of right endpoints of intervals stored at $v$;
    \item $\sum_{I \in v} I^-$, the sum of left endpoints of intervals stored at $v$.
\end{itemize}
Given a query interval $J$, we obtain $O(\log^2 n)$ canonical set of nodes in \ITSUM,
where each node $v$ falls into: 
(1)~every $I \in v$ completely covers $J$;
(2)~$J$ completely covers every $I \in v$;
(3)~every $I \in v$ partially intersects $J$ with $I^+ \in J$;
(4)~every $I \in v$ partially intersects $J$ with $I^- \in J$.
Then, we can rewrite the \SUM-durability of intervals with respect to $J$ as follows:
\begin{displaymath}
\sum_{I \in \I} |I \cap J|
= \sum_v \sum_{I \in v} |I \cap J|
= \sum_v \left\{
\begin{array}{ll}
    |v| \cdot |J| & \textrm{if (1);}\\
    \sum_{I \in v} |I| & \textrm{if (2);}\\
    \sum_{I \in v} I^+ - |v| \cdot J^- & \textrm{if (3);} \\
    |v| \cdot J^+ - \sum_{I \in v} I^- & \textrm{if (4).}
\end{array}
\right.
\end{displaymath}
Note that \ITSUM\ can be constructed in $O(n\log^2 n)$ time and uses $O(n\log n)$ space.
This way, we have a procedure \textsc{ComputeSumD} which, given \ITSUM\ and $J$, returns $\sum_{I\in \I}|I\cap J|$ in $O(\log^2 n)$ time. The interval tree $\ITSUM$ can also be used to find $\C_p$.

\paragraph{Data structure.}
While \ITSUM\ makes it efficient to sum durabilities over a set of intervals given $I_p \cap I_q$ for a candidate pair $(p,q)$,
we cannot afford to check all possible pairs,
and we have not yet addressed the challenge of obtaining the intervals of interest
(which must come from witness points incident to both $p$ and $q$) in the first place.
The high-level idea is to leverage the same space decomposition from the previous sections
to efficiently obtain canonical subsets of witness points, in their implicit representation.
These canonical subsets of intervals serve as the basis for building \ITSUM\ structures.
In more detail, we construct \DTCSUM\ in a similar way as \DTC\ in Section~\ref{sec:tau-durable-ball}.
Like \DTC, \DTCSUM\ is a two-level data structure consisting of a cover tree and an interval tree variant (as described above) for every node of the cover tree.
For each cover tree node $u$, let $\C_u$ denote the subset of points in $P$ within the ball of $u$ centered at $\rep_u$.
We build $\ITSUM_u$ over $\C_u$ with \SUM\ annotations, and
for each node in $\ITSUM_u$, we also store points in decreasing order of their right interval endpoints.
Overall, we can construct \DTCSUM\ in $O(n\log^3 n)$ time having $O(n\log^2 n)$ space.

\begin{algorithm}[t]
\caption{{\sc ReportSUMPair}$(\DTCSUM, p, \tau, \eps)$}
\label{alg:sum}

    $\C_{p}: \{\C_{p,1}, \C_{p,2}, \cdots, \C_{p,k}\} \gets \DTCQ(p,\tau,\eps/2)$,
    with $\rep_i$ as the representative point of the cover tree node for $\C_{p,i}$,
    and $\ITSUM_{p,i}$ as the annotated interval tree for the cover tree node\;
\ForEach{$j \in [k]$}{
    \ForEach{$q \in \C_{p,j}$ in descending order of $I_q^+$}{
        $t \gets 0$\;
        \ForEach{$i \in [k]$}{
            \If{$\dist(\rep_i, \rep_j) \le 1 + \frac{\eps}{2}$}{
                $t \gets t + \textsc{ComputeSumD}(\ITSUM_{p,i}, I_p \cap I_q)$\label{line:SUMQuery}\;
            }
        }
        \lIf{$t\geq \tau + 2\cdot|I_p\cap I_q|$}{
            \KwSty{report} $(p,q)$
        }\lElse{
            \KwSty{break}%
        }
    }
}
\end{algorithm}

\paragraph{Algorithm.}
We report all $\tau$-\SUM-durable pairs $(p,q)$ where $I_p^- \ge I_q^-$ (to avoid duplicates); we say $p$ \emph{anchors} the pair.
For each $p \in P$, we invoke \textsc{ReportSUMPair} (Algorithm~\ref{alg:sum}) to report $\tau$-\SUM-durable $\eps$-pairs $(p,q)$ anchored by $p$.
To this end, \textsc{ReportSUMPair} runs the $\tau$-durable ball query $\DTCQ(p,\tau,\eps/2)$ over \DTCSUM,
and obtain a family of result point sets $\C_{p,1}, \C_{p,2}, \ldots, \C_{p,k}$ for some $k = O(\eps^{-\ddim})$.
Each $\C_{p,j}$ is covered by a cover tree ball in $\DTCSUM$ with diameter of no more than $\eps/2$,
and contains all points $q$ within the ball where $I_q^-+\tau \le I_p^-+\tau \le I_q^+$, as explained in Section~\ref{sec:tau-durable-ball}, sorted with respect to $I_q^+$.
Let $\ITSUM_{p,j}$ denote the interval tree for the cover tree node corresponding to $\C_{p,j}$.
For each $j$, we go through each point $q \in \C_{p,j}$ in decreasing order of right endpoints to check whether $(p,q)$ is $\tau$-\SUM-durable.
To do this check, we consider witnesses from point sets $\C_{p,1}, \C_{p,2}, \ldots, \C_{p,k}$.
We can skip an entire set $\C_{p,i}$ if its ball center $\rep_i$ is too far from $\rep_j$, because all points in $\C_{p,i}$ would be too far from $q$.
Otherwise, we query $\ITSUM_{p,i}$ using interval $I_p \cap I_q$ to obtain the sum of durabilities over all witnesses in $\C_{p,i}$.
We compute these partial sums together and compare the total with $\tau + 2\cdot|I_p\cap I_q|$
(note that the second term accounts for the fact that the partial sums include the contributions of $p$ and $q$ themselves,
which should be discounted).
If the total passes the threshold, we report $(p,q)$.
If not, we stop consider any remaining point $q' \in \C_{p,j}$ (which has $I^+_{q'} < I^+_q$),
since $I_p \cap I_{q'} \subset I_p \cap I_q$ and will surely yield a lower total durability. This is the key for output-sensitive time.

\paragraph{Correctness.} First, each pair $(p,q)$ is reported at most once, as $(p,q)$ is reported if $I_p^-\geq I_q^-$. Next, we show that every pair reported must be a $\tau $-SUM-durable $\eps$-triangle. Consider a pair $(p,q)$ that is reported. Note that our algorithm considers a node $u_i$ from \DTCSUM\ with radius $\eps/4$ if and only if $\dist(p, \rep_i)\leq 1+\eps/4$.
If $q\in u_i$,  we have $\dist(p,q)\leq \dist(p, \rep_i)+\dist(\rep_i,q)\leq 1+\eps/2$. Hence, in any pair $(p,q)$ we return it holds that $\dist(p,q)\leq 1+\eps$.
Then we only consider points within distance $1+\eps$ from both $p,q$ to find the sum of their corresponding intervals. Indeed, we only consider the pairs $\C_{p,i}, \C_{p,j}$ with $\dist(\rep_i, \rep_j)\leq 1+\eps/2$. Let $q'$ be any point from $\C_{p,i}$. We have $\dist(p,q')\leq 1+\eps/2$, and $\dist(q,q')\leq \dist(\rep_i, q)+\dist(\rep_i, \rep_j)+\dist(\rep_j, q')\leq 1+\eps$. Overall, by showing i) $\dist(p,q)\leq 1$, ii) that we only take the sum of intervals in $I_p\cap I_q$ among points (witness points) within distance $1+\eps$ from both $p, q$, and iii) the correctness of the \ITSUM\ data structure, we conclude that the reporting pair $(p,q)$ is a $\tau$-durable $\eps$-pair.

Finally, we show that every $\tau$-SUM-durable pair will be reported. Let $(p,q)$ be an arbitrary $\tau$-SUM-durable pair. Suppose $q\in u_j$, where $u_j$ is a node of \DTCSUM\ of radius at most $\eps/4$, with representative point $\rep_j$. Since $\dist(p, \rep_j)\leq \dist(p,q)+\dist(q,\rep_j)\leq 1+\eps/4$, we have that $u_j \in \mathcal{C}_p$. Without loss of generality, assume that $q\in \C_{p,j}$. Next, we show that for point $q' \in P$, if $\dist(p,q') \leq 1$ and $\dist(q,q')\leq 1$, we always consider $q'$ in the witness set.
Since $\dist(p,q') \leq 1$ we have that $q'\in \C_p$.
Without loss of generality, assume that $q'\in \C_{p,i}$.
In this case,  $\dist(\rep_j,\rep_i)\leq \dist(\rep_j,q)+\dist(q,q')+\dist(q',\rep_i)\leq 1+\eps/2$, so $q'$ is included in $\ITSUM_{p,i}$ considered in line~\ref{line:SUMQuery} of Algorithm~\ref{alg:sum}.  It remains to show that if $(p,q)$ is a $\tau$-durable pair, $q$ must be visited during the traversal of points in $\C_{p,j}$. 
We prove it by contradiction. Let $w \in \C_{p,j}$ be a point such that $I_w^-\leq I_p^- \le I_q^+ \le I_w^+$. Suppose after visiting $w$, the traversal of points in $\C_{p,j}$ stops. Implied by the stopping condition, $(p,w)$ is not a $\tau$-SUM-durable pair. Meanwhile, as $(p,q)$ is a $\tau$-SUM-durable pair, $(p,w)$ must also be a $\tau$-SUM-durable pair, implied by $I_q\cap I_p \subseteq I_w\cap I_p$, and the fact that we run the $\textsc{ComputeSumD}(\ITSUM_{p,i}, I_p \cap I_q)$ query on the same sets $\C_{p,i}$,
coming to a contradiction. Thus, every $\tau$-SUM-durable pair must be reported.

\paragraph{Time Complexity.} The construction time of $\ITSUM$ is $O(n\log^2 n)$, so it takes $O(n\log^3 n)$ time to construct \DTCSUM.
For each $p$, it takes $O(\eps^{-O(\ddim)}+\log n)$ time to derive the canonical set of nodes $\mathcal{C}_p$ and $O(\eps^{-O(\ddim)}\log n)$ time to derive the sorted intervals in every node of $\mathcal{C}_p$. For each (sorted) interval $I_q$ in $\C_{p,i}$ we visit $O(\eps^{-O(\ddim)})$ other nodes $\C_{p,j}$ and we run a $O(\log^2 n)$ time query to find the sum using \ITSUM. When we find out that $q$ does not form a $\tau$-SUM-durable pair with $p$ we skip the rest points in $\C_{p,i}$ so the running time is output-sensitive. Overall, the running time is bounded by $O(n\log^3 n + (n+\OUT) \cdot \eps^{-O(\ddim)}\log^2 n)$,
where $K_\tau\leq \OUTS \leq K_{\tau}^{\eps}$.
\begin{theorem}
    \label{thm:offlineSUM}
    Given $(P,\dist, I)$, $\tau > 0$ and $\eps > 0$,
    the $\eps$-approximate \problemthree-\SUM\ problem can be solved in
    $O(n\log^3 n+(n+\OUT)\cdot \eps^{-O(\ddim)}\log^2 n)$ time,
    where $n = |P|$, $\ddim$ is the doubling dimension of $P$, and $\OUT$ is the number of pairs reported.
\end{theorem}
\subsection{\UNION}
\label{sec:aggregate:union}

\newchange{
Solving the general \problemthree-\UNION\ problem is challenging because of the inherent hardness of computing the union of intervals that intersect a query interval, i.e., we cannot design an efficient primitive for \UNION\ as the \textsc{ComputeSumD} primitive for \SUM.
In practice, even if the size of the witness set $U$ is large,
a smaller subset of $U$ may be all that is required for its union to reach the durability parameter.
With this observation, we approach the problem by designing an algorithm whose performance depends on $\kappa$, a constraint on the size of the witness set.
More precisely, given a durability parameter $\tau > 0$ and a positive integer $\kappa \in \mathbb{Z^+}$,
we say a pair $(p_1, p_2) \in P \times P$ is {\em $(\tau,\kappa)$-\UNION-durable} if
$\dist(p_1, p_2) \leq 1$ and
there exists $U \subseteq \{ u \in P \mid \dist(p_1,u), \dist(p_2, u) \le 1 \}$ such that
$|U| \le \kappa$ and $\left|\bigcup_{u\in U} I(u, {p_1} ,{p_2})\right| \ge \tau$.
An approximate version is defined by replacing the distance constraint $\le 1$ with $\le 1 + \eps$.
We present an $\O((n+\OUT)\cdot \kappa\eps^{-O(\ddim)})$-time algorithm for $\OUT \in \left[|K_{\tau, k}|, \big|K_{(1-1/e)\tau, \kappa}^{\eps}\big|\right]$,
with $K_{\tau,\kappa}$ denoting the set of $(\tau,\kappa)$-\UNION\ durable pairs
and $K_{(1-1/e)\tau, \kappa}^{\eps}$ denoting the set of $((1-1/e)\tau,\kappa)$-\UNION\ durable $\eps$-pairs.
}

\paragraph{$(\tau,\kappa)$-\UNION-durable pair.}
Next, we focus on finding $(\tau,\kappa)$-\UNION-durable pairs for some known $\kappa$,
which should work well in practical cases where a handful of witness points are able to provide sufficient coverage for the pair.
The overall algorithm has the high-level idea of leveraging the space decomposition as the \problemthree-\SUM\ case in Section~\ref{sec:aggregate:sum},
but it requires a primitive different from \textsc{ComputeSumD} and a different way of invoking this primitive across witness subsets.

\paragraph{High-level Idea.} Given a set of intervals $\I$ and a target interval $J$,
our approach is to find a subset of intervals $X\subseteq \mathcal{I}$ of size $\kappa$
that maximizes the \UNION-durability with respect to $J$,
namely $\big|\bigcup_{I \in X} (I\cap J)\big|$ (to compare with $\tau$).
There is an apparent connection to the {\em maximum $\kappa$-coverage} problem,
where given a family of sets over a set of elements, we want to choose $\kappa$ sets to cover the maximum number of elements.
Here, we can regard each set as $I \cap J$ for each $I \in \I$,
and the goal is to choose $\kappa$ such sets to cover as much of $J$ as possible.
The standard greedy algorithm gives an $(1-1/e)$-approximation for this problem~\cite{hochbaum1996approximating},
which inspires us to follow a similar approach. We leave details on data structures, pseudocode, correctness and complexity analysis to Appendix~\ref{appendix:aggregate}.
Our greedy approach chooses one interval at time to cover $J$,
and the choice is always the one that maximizes the resulting increase in coverage.
In more detail, let $X \subseteq \I$ denote the set of intervals already chosen,
which leaves $J \setminus \bigcup_{I \in X} I$, the uncovered parts of $J$, as a set $Y$ of intervals.
Consider the pair $(I_x, I_y)$, where $I_x \in \I \setminus X$, and $I_y \in Y$, with the largest overlap, i.e., $I_x \cap I_y$;
we greedily choose $I_x$ as the next interval to cover $J$.

To implement this greedy approach efficiently,
we build a data structure \DTCUNION\ similarly as \DTCSUM\ in Section~\ref{sec:aggregate:sum}.
\DTCUNION\ uses a different variant of the interval tree \ITUNION,
which, given a query interval $J$, finds the indexed interval with the largest overlap with $J$.
The overall algorithm reports all $(\tau,\kappa)$-\UNION\ durable pairs for each anchor point $p$ by querying \DTCUNION,
and for each candidate $(p,q)$, performs the greedy choice $k$ times to compute the \UNION-durability of $(p,q)$.
Each greedy choice involves querying the \ITUNION\ structures for the $O(\eps^{-O(\ddim)}\log^2 n)$ canonical subsets of witness points;
some additional elementary data structures help ensure that the greedy algorithm takes $O\left(\kappa\eps^{-O(\ddim)}\log^2 n\right)$ time. In Appendix~\ref{appendix:aggregate}, we show that the overall time is $O\left(n\log^{3}n + (n+\OUT)\cdot \eps^{-O(\ddim)} \kappa\log^2 n\right)$, where $|K_{\tau,\kappa}| \leq \OUT \leq \left|K_{(1-1/e)\tau, \kappa}^{\eps}\right|$. Putting everything together, we obtain:
\begin{theorem}
    \label{thm:offlineUNION}
    Given $(P,\dist, I)$, $\tau > 0$, $\eps > 0$, and integer $\kappa \in \mathbb{Z}^+$,
    the $\eps$-approximate \problemthree-\UNION\ problem can be solved in
    $O(n\log^{3}n + (n+\OUT)\cdot \eps^{-O(\ddim)} \kappa\log^2 n)$ time,
    where $n = |P|$, $\ddim$ is the doubling dimension of $P$, and $\OUT$ is the number of pairs reported.
\end{theorem}
\section{Related Work}

In database and data mining, there is a large body of literature on finding patterns in temporal graphs~\cite{araujo2016discovery, gong2012community, lin2008facetnet, franzke2018pattern, semertzidis2016durable, yang2016diversified, redmond2016subgraph}.
Hu et al.~\cite{hu2022computing} studied the problem of computing temporal join queries efficiently;
the problem of finding durable triangles is a special case of the problem they studied with self-joins.
While~\cite{hu2022computing, yang2016diversified} have provable guarantees,
the algorithms are expensive, requiring time super-linear in the number of edges to report all durable triangles.
Recently, Deng et al.~\cite{deng2023space} proposed algorithms to report or count triangles (and other simple patterns) in time super-linear in the graph size.
In contrast, we work with an implicit representation of the proximity graph
and design algorithms that run in time near-linear in the number of nodes and output size.

There is another line of work in computational geometry on detecting triangles and other simple patterns in intersection graphs.
Eppstein and Erickson~\cite{eppstein1994iterated} gave an $O(n\log n)$ algorithm to detect
if an intersection graph consisting of unit balls in $\mathbb{R}^d$ has a constant clique. 
Kaplan et al.~\cite{kaplan2019triangles} can detect a triangle in a unit-disk graph in $\mathbb{R}^2$ in $O(n\log n)$ time
where edges can be weighted.
The approach in~\cite{chan2013klee} can detect in $\O(n^{d/2})$ time
if a clique of constant size exists in an intersection graph of general boxes in $\mathbb{R}^d$.
Chan~\cite{chan2022finding} recently improved the results on detecting cliques, cycles, and other simple patterns in intersection graphs,
where the nodes are boxes, general fat objects in $\mathbb{R}^d$, or segments in $\mathbb{R}^2$,
and two nodes are connected if the corresponding objects intersect.
For example, if nodes are fat objects, their algorithm can detect a constant cycle or clique in $O(n\log n)$ time.
The problems we focus on in this paper have major differences with this line of work:
(\romannumeral 1) previous methods only worked for detecting whether a pattern exists, while our goal is to report all patterns;
(\romannumeral 2) all previous works focused on non-temporal graphs, while we consider the more challenging {\em temporal graphs}, where nodes have lifespans;
(\romannumeral 3) we additionally considered an incremental reporting setting to support queries with different parameters.

The notion of durability has been studied in other queries,
such as durable top-$k$ queries~\cite{gao2018durable, gao2021durable} and durability prediction~\cite{gao2021efficiently}.
It also has been studied in computational topology, where the goal is to compute ``persistent'' (durable) topological features; see~\cite{edelsbrunner2022computational, dey2022computational}.
\section{Conclusion}

In this paper, we have studied the problem of reporting durable patterns in proximity graphs.
We work with an implicit representation of the input graph,
and propose efficient algorithms that run in near-linear time in the number of nodes,
under any general metric with bounded doubling dimension.
For future work, we believe that some of our algorithms and data structures can also be used for
counting 
durable patterns in near-linear time (instead of reporting them).
Second, while we have focused on simple patterns such as triangles and paths,
it would be interesting to explore near-linear time algorithms for more general and complex patterns.
Third, we have considered only the case when nodes have lifespans but otherwise remain stationary;
one could further consider the case when their positions change over time (hence inducing also lifespans on edges).
A possible direction is to use \emph{kinetic data structures} to maintain the evolving graph topology.
Finally, a challenging question is whether we can extend our approach to a general graph already with an explicit representation,
but without first computing an embedding.

\newpage
\bibliographystyle{abbrv}
\bibliography{ref, acmart}
\newpage
\appendix


\section{Cover tree for ball reporting queries}
\label{appndx:covertree}
We consider the case where the doubling dimension is constant or the expansion constant is bounded by a constant. Furthermore, we assume that the spread of the items $P$ is bounded by a polynomial on $n$.
Given a query item $q$ and an error threshold $\eps$ the goal is to find a family of sets $C=\{C_1, \ldots, C_m\}$, with $m=O(\eps^{-O(\ddim)})$, such that $C_i\subseteq P$, $C_i\cap C_j=\emptyset$, for every item $p\in P$ with $\dist(p,q)\leq 1$, $p\in C_j$, for an index $j\leq m$, and for every point $p\in \bigcup_{i\leq m}C_i$ it holds that $\dist(p,q)\leq 1+\eps$. Finally we require that the distance of any pair of items inside $C_i$ to be at most $\eps$.


When the spread is bounded, the cover tree consists of $O(\log n)$ levels. Let assume that the root has the highest level and the leaf nodes the lowest level. Each node $v$ in the cover tree is associated with a representative point $\rep_v\in P$.
For each node $v$ in level $i$ of the cover tree it holds that:
i) If $u$ is another node in level $i$ then $\dist(\rep_v, \rep_u)>2^i$.
ii) If $v$ is not the root node, it always has a parent $w$ in level $i+1$. It holds that $\dist(\rep_v, \rep_w)<2^{i+1}$.
iii) If $v$ is not a leaf node, $v$ has always a child $w$ such that $\rep_v=\rep_w$.
Assuming that the doubling dimension is $\ddim$ we have that each node $v$ of the cover tree has $O(2^{O(\ddim)})=O(1)$ children. The same, constant bound, holds for bounded expansion constant.
The standard cover tree has space $O(n)$ and can be constructed in $O(n\log n)$ time~\cite{beygelzimer2006cover, har2005fast}.

Notice that every node in the lowest level contains one item from $P$ and each item in $P$ appears in one leaf node. Let $P_v$ be the set of points stored in (the leaf nodes of) the subtree rooted at node $v$. We do not explicitly store $P_v$ in every node $v$ of the cover tree. Instead, for every node $v$ we add a pointer to the leftmost leaf node in the subtree rooted at $v$. If we also link all the leaf nodes, given a node $v$, we can report all points in $P_v$ following the pointers, in $O(|P_v|)$ time.
Our modified cover tree has space $O(n)$ and can be constructed in $O(n\log n)$ time.
For each node $v$ in level $i$, let $r_v=2^i$ be its separating radius and $e_v=2^{i+1}$ be its covering radius.

\begin{lemma}
If $p\in P_v$ then $\dist(p,\rep_v)< e_v$.
\end{lemma}
\begin{proof}
It follows by induction on the level of the tree. In the leaf nodes it holds trivially. We assume that it holds for all nodes in level $i-1$. We show that it holds for every node at level $i$. Let $v$ be a node at level $i$. By definition we have that if $w$ is a child of $v$ then $\dist(\rep_v, \rep_w)<r_v$. By the induction assumption, if $p\in P_w$ it holds that $\dist(p,\rep_w)< e_w$, so
$\dist(p,\rep_v)\leq \dist(\rep_v, \rep_w)+\dist(p,\rep_w)<r_v+e_w=2^{i+1}=e_v$.
\end{proof}

\paragraph{Query procedure}
Given a query point $q$ and an error threshold $\eps$, we start the query procedure in the modified cover tree we constructed above. In each level $i$ we visit the nodes $v$ such that $r_v>1$ and $\dist(q,\rep_v)\leq 1+e_v$. Let $V_i$ be the nodes in level $i$ we visit such that $r_v=1$ for $v\in V_i$.
Then we consider each of the node $v\in V_i$ and we get all nodes $u$ in the subtree of $v$ with $r_u=\eps/4$. Let $C'$ be the set of all nodes $u$ we found. We go through each node $w\in C'$ and we check whether $\dist(q,\rep_w)
\leq 1+\eps/2$. If yes, then we add $P_w$ in $C$. Otherwise, we skip it.

\begin{lemma}
The query procedure is correct.
\end{lemma}
\begin{proof}
Let $p\in P$ be an item such that $\dist(q,p)\leq 1$. We need to show that $p$ belongs in a set in $C$. Let $i$ be the level of the node $v$ such that $p\in P_v$ and $r_v=1$. Since $\dist(q,p)\leq 1$ it also holds that $\dist(q,\rep_v)\leq \dist(q,p)+\dist(p,\rep_v)\leq 1+e_v$. Hence, we will visit node $v$ in the query procedure and we will add it in set $V_i$.
Since $p\in P_v$ and $v\in V_i$, by definition, item $p$ lies in one of the nodes $w$ in $C'$. We have 
$\dist(q,\rep_w)\leq \dist(q,p)+\dist(p,\rep_w)\leq 1+\eps/2$. So we will keep $w$ in $C$.

Finally, notice that for each $p\in P_w$ for a node $w$ in $C$, we have  $\dist(q,p)\leq \dist(q,\rep_w)+\dist(p,\rep_w)\leq 1+\eps/2+\eps/2\leq 1+\eps$.
So the query procedure is correct.
\end{proof}

\begin{lemma}
The query procedure runs in $O(\log n+ \eps^{-O(\ddim)})$ time.
\end{lemma}
\begin{proof}
We first bound the number of nodes $v$ we visit in level $i$ with $r_v=1$. Recall that we only consider $v$ if $\dist(q,\rep_v)\leq 1+e_v\leq 3$. Equivalently, we can think of a ball $\mathcal{B}$ of radius $3$ and center $q$. Each node $v$ defines a ball $\mathcal{B}_v$ with center $\rep_v$ and radius $1$. Also notice that the centers of any two balls $\mathcal{B}_v, \mathcal{B}_u$ have distance at least $1$. The number of nodes we visit in level $i$ is the same as the number of balls $\mathcal{B}_v$ that intersect $\mathcal{B}$.
Using the bounded doubling dimension, it is easy to argue that ball $\mathcal{B}$ of radius $3$ can be covered by at most $O(2^{O(\ddim)})=O(1)$ balls of radius $1$ (a similar argument holds for bounded expansion constant). Hence, we can argue that in each level above $i$ we only visit $O(1)$ number of nodes. So in total we visit $O(\log n)$ nodes until we reach level $i$ with $2^i=1$. Next, since the doubling dimension is $\ddim$ each node can have $O(2^{O(\ddim)})$ children. It follows that the number of nodes we visit with $r_u=\eps/4$ in the subtree of any node $v$ (with $r_v=1$) in $C'$ is $O(\eps^{-O(\ddim)})$. Overall, the query time is $O(\log n + \eps^{-O(\ddim)})$.
\end{proof}

\section{Exact algorithms and data structures for $\ell_\infty$}
\label{appndx:linfty}
\subsection{$\tau$-durable ball query}
\label{sec:tau-durable-range}

\begin{definition}[$\tau$-durable range query]
 Given a set of items $P \subseteq \Re^d$, a durability parameter $\tau > 0$, a query point $p \in P$ and a rectangle $R$, the $\tau$-durable range query asks to find all points $q\in P\cap R$ such that $|I_p\cap I_q| \geq \tau$ and $I^-_p \in I_q$. 
\end{definition}
Intuitively, this query asks to find all points in $R$ whose intervals trimmed by $\tau$ from the left endpoint intersect $I_p^-+\tau$. This query can be solved by a multi-level data structure built on range tree (or cover tree) and interval tree. In particular, we construct a range tree $\RT$ over the set of points in $P$. For each node $u$ in the $d$-th level of $\RT$, we further construct an interval tree $\IT_u$ over the intervals of the points lying in the rectangle defined by $u$. Let $\DTR$ be the data structured constructed above. By slightly abusing the notation, let $\DTRQ(p, \tau, R)$ be the query procedure with parameters $p, \tau, R$.  The data structure $\DTR$ uses $O(n\log^{d+1} n)$ space and can be constructed in $O(n \cdot \log^d n)$ time. The query $\DTRQ(p, \tau, R)$ can be answered over $\DTR$ as follows. 
    {\bf (Step 1)} we first query the $\RT$ with the input rectangle $R$, and will obtain a set of $d$-th level nodes that lie inside $R$;
    {\bf (Step 2)} for each node $u$ returned, we query the $\IT_u$ with $I^-_p$, and will obtain all $q \in \IT_u$, such that $I_q^- +\tau\leq I_p^- +\tau\leq I_q^+$. 
    By resorting to existing result on range reporting, we obtain:

\begin{lemma}
\label{lem:helper-1}
    A data structure can be built in $O(n \cdot \log^d n)$ time with $O(n\log^{d+1} n)$ space, such that any $\tau$-durable range query can be answered in $O(\log^{d+1} n + \OUT)$ time, where $\OUT$ is the number of query results. 
\end{lemma}

 Using the same data structure, we can also determine whether $K>0$ or not, i.e., if there exists any query answer for a $\tau$-durable range query, in $O(\log^{d+1} n)$ time.

\subsection{\problemone}
\label{subsec:unitdisk}
Assume $G$ is a unit disk graph under the $\ell_\infty$-norm distance function. Recall that $P \in \Re^d$ is a set of points, where each point $p$ is associated with the lifespan interval $I_p$. 

\begin{figure}
    \centering
    \includegraphics[scale=0.9]{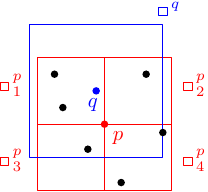}
    \caption{An illustration of Algorithm~\ref{alg:report-unit}: $p$ is visited.}
    \label{fig:offline}
\end{figure}

\begin{algorithm}[t]
\caption{{\sc ReportTriangle-I}$(p, \tau, \DTR)$}
\label{alg:report-unit}

$\square^p \gets \{\square^p_1, \square^p_2, \ldots, \square^p_{2^d}\}$ as canonical unit squares of $p$\;
\ForEach{$j \in [2^d]$}{
    $\mathcal{C}_{p,j} \gets \DTRQ\left(p, \tau, \square^p_{j}\right)$\;
    Report $(p,q,q')$ for every $(q,q') \in \mathcal{C}_{p,j} \times \mathcal{C}_{p,j}$ if $q \neq q'$\;
    \ForEach{$q \in \mathcal{C}_{p,j}$}{
        \ForEach{$j < k \le 2^d$}{
            $\mathcal{L}_{p,q} \gets \DTRQ\left(p, \tau, \square^q \cap \square^p_{k}\right)$\;
            Report $(p,q,q')$ for every $q' \in \mathcal{L}_{p,q}$\;
        }
    }
}

\end{algorithm}

\paragraph{Algorithm.} As a preprocessing step, we construct the data structure $\DTR$ as described in Section~\ref{sec:tau-durable-range} over $P$. We visit every point in $P$. For each point $p$ visited, we enumerate all $\tau$-durable triangles that contain $p$ such that $\max\{I_q^-, I_{q'}^-\} + \tau\leq I_p^- +\tau\leq \min\{I_q^+, I_{q'}^+\}$ for every $(p,q,q')$ reported. Below, we focus on one specific point $p$, as described in Algorithm~\ref{alg:report-unit}.

We construct a square with side-length $2$ and center $p$, and we split it into $2^d$ canonical unit squares $\square^p= \left\{\square^p_1, \ldots, \square^p_{2^d}\right\}$.
Then we consider each $\square^p_j$ in order, by running a $\tau$-durable range query $\DTRQ\left(p, \tau, \square^p_j\right)$. Let $\C_{p,j}$ be the set of canonical subsets returned. We also use the same notation $\C_{p,j}$ for the set of points in the corresponding canonical subset. We further distinguish a $\tau$-durable triangle $(p,q,q')$ into two types: (1) $q,q' \in \square^q_j$; (2) $q \in \square^q_j, q' \in \square^q_k$ for $j \neq k$. We only need to report triangles in type (1) when $|\C_{p,j}|>1$, i.e., there are multiple points from $\square^p_j$ within durability threshold, hence every pair $(q,q') \in \C_{p,j} \times \C_{p,j}$ with $q \neq q'$ must form a $\tau$-durable triangle with $p$. To report triangles in type (2), for each point $q \in \C_{p,j}$ returned, we issue another $\tau$-durable range query $\DTRQ\left(p, \tau, \square^q\cap \square^p_{k}\right)$ for every $k \ge j$, whose answer is exactly the set of points that form a $\tau$-durable triangle in type (2) with $p,q$.


\mparagraph{Correctness}
It is straightforward to see that if the union of all $\tau$-durable triangles reported for each point $p$ that intersect $I_p^-$, in the end, it is exactly the set of all $\tau$-durable triangles. Furthermore, each $\tau$-durable triangle is reported once. Finally, it is easy to show that all $\tau$-durable triangles are reported with respect to a point $p$ whose intervals intersect $I_p^-$. When reporting a triangle $(p,q,q')$, we always have $||p-q||\leq 1$, $||p-q'||\leq 1$, $||q-q'||\leq 1$ and $|I_p\cap I_q\cap I_{q'}|\geq \tau$ due to the definition of the unit squares $\square^p$ and data structure $\DTR$.

\mparagraph{Time Complexity} Implied by Lemma~\ref{lem:helper-1}, the data structure $\DTR$ can be constructed in $O(n\log^{d+1} n)$ time, and each $\tau$-durable range query can be answered in 
$O(\log^{d+1} n)$ time. Then, we count the number of such range queries issued in total. For each point $p$, it invokes $O(2^d)$ queries each for one distinct unit square, and $O(|\C_{p,j}|)$ queries each for one point $q \in \C_{p,j}$. We point out two critical observations: 
\begin{itemize}[leftmargin=*]
    \item there are at most $2^d$ different $j$ such that $|\C_{p,j}|=1$;
    \item if $|\C_{p,j}| \ge 2$, then all points in $\C_{p,j}$ form $\frac{1}{2} \cdot |\C_{p,j}| \cdot (|\C_{p,j}|-1)$ $\tau$-durable triangles together with $p$.
\end{itemize}
Combining these two observations yields the following inequality:
\begin{align*}
    \sum_{j \in [2^d]} |\C_{p,j}| & \le \sum_{j \in [2^d]: |\C_{p,j}|=1} 1 + \sum_{j \in [2^d]: |\C_{p,j}|\ge 2} |\C_{p,j}| \\
    & \le 2^d + \sqrt{2 \cdot |T_{p, \tau}|} + 1 = O(\sqrt{|T^p_{\tau}|})
\end{align*}
where $T^p_{\tau}$ is number of $\tau$-durable triangles that contain $p$.
Then, the total number of invocations would be upper bounded by 
\begin{align*}
    \sum_{p \in P} \left(2^d+ \sum_{j \in 2^d} |\C_{p,j}|\right) = &  n + \sum_{p \in P} \sqrt{|T^p_\tau|}  
    \le n + \sqrt{n \cdot |T_\tau|} 
    \le n + |T_\tau|
\end{align*}
where the second inequality follows Cauchy-Schwarz inequality.
Thus, the total running time is $O\left((n + |T_\tau|) \cdot \log^{d+1} n\right)$. We can further improve the analysis, by a fine-grained analysis: 
\begin{itemize}[leftmargin=*]
    \item If $n \leq |T_\tau|\leq n\log^{2d+2} n$, then $\sqrt{n \cdot |T_\tau|} \cdot \log^{d+1} n \le n \log^{2d+2} n$; 
    \item If $n \ge |T_\tau|$, then $\sqrt{n \cdot |T_\tau|} \cdot \log^{d+1} \le n \log^{d+1}n$;
    \item If $|T_\tau|\geq n\log^{2d+2} n$, then $\sqrt{n \cdot |T_\tau|} \cdot \log^{d+1} n\le n \log^{d+1}n$;
\end{itemize} 
Overall, the running time can be upper bounded by $O(n\log^{2d+2}n + |T_\tau|)$.
Compared with our previous result $O((n+|T_\tau|) \cdot \log^{d+1}n)$, we have removed the logarithmic factor from $|T_\tau|$, but increased the logarithmic factor on the input size $n$. This is significantly important, since $|T_\tau|$ can be as large as $O(n^3)$ in the worst case, while the increase on the logarithmic factor can be negligible.

\begin{theorem}
    \label{thm:offline}
Given $(P,\phi,I)$, and $\tau>0$, where $P$ is a set of $n$ points in $\Re^d$ and $\phi$ is the $\ell_{\infty}$ norm,
 \problemone\ can be solved in
    $O(n \log^{2d+2} n + |T_\tau|)$ time,
    where $T_\tau$ is the set of $\tau$-durable triangles formed by points in $P$.
\end{theorem}

\subsection{\problemtwo\ for the $\ell_\infty$ metric}
\label{sec:linfty-reporting}

\paragraph{Compute Activation Thresholds. } 
For a set $P$ of $n$ points, we first build the data structure $\DTR$ as described in Section~\ref{sec:tau-durable-range} over $P$, and then invoke Algorithm~\ref{alg:detect}.
In particular, we construct an extended version of $\DTR$ to handle interval queries with one additional linear constraint on the endpoints as we will see later. Hence, the query time in $\DTR$ is increased by $\log n$.
%
More specifically, we construct a square with side-length $2$ centered at $p$, and split it into $2^d$ canonical unit squares $\square^p=\{\square^p_1, \ldots, \square^p_{2^d}\}$.
Then we consider each $\square^p_j$ in order, by running a $\tau$-durable range query $\DTRQ\left(p, \tau, \square^p_j\right)$.  Let $\C_{p,j}$ be the set of points returned. More specifically, using the extended version of $\DTR$, we further partition $\C_{p,j}$ into two different subsets of points,  $\NewerC_{p,j}$ and $\OlderC_{p,j}$. Intuitively, for every point $q$ returned from $\C_{p,j}$, $\tau' \le |I_q \cap I_p| < \tau$ if $q \in \NewerC_{p,j}$, and $|I_q \cap I_p| \ge \tau$ otherwise. We can find $O(\log^{d+2} n)$ canonical subsets containing $\NewerC_{p,j}, \OlderC_{p,j}$ using $\DTR$.
If there is a unit square $\square_j^p$ that contains at least 2 points and at least one of them is in $\NewerC_{p,j}$ then we return that there exists a $\tau'$-durable triangle that is not $\tau$-durable.
Otherwise, similar to Algorithm~\ref{alg:report-unit}, we check if any point of $\NewerC_{p,j}$ forms a $\tau'$-durable triangle, which is not $\tau$-durable, with any other point in $C_{p,j}$. If this is the case then we return true. Otherwise, we return false. 
As mentioned, we perform at most $O(\log n)$ guesses and it takes $O(\log^{d+2}n)$ time to run Algorithm~\ref{alg:detect} for each guess of $\tau'$.
%
%
\renewcommand{\S}{\mathcal{S}}
%


\begin{algorithm}[t]
\caption{{\sc DetectTriangle-I}$(p, \DTR, \tau, \tau')$}
\label{alg:detect}
$\square^p \gets \{\square^p_1, \square^p_2, \ldots, \square^p_{2^d}\}$ as canonical unit squares of $p$\;
\ForEach{$j \in [2^d]$}{
    $\mathcal{C}_{p,j} \gets \DTRQ(p, \tau, \square^p_{j})$\;
    $\NewerC_{p,j} = 
    \left\{q \in \C_{p,j}: I^-_q \le I_p^- \textrm{ and } I_p^- + \tau' \le I^+_q < I_p^- +\tau\right\}$\;
    $\OlderC_{p,j} = \left\{q \in \C_{p,j}: I_q^-\leq I_p^- \textrm{ and } I_q^+\geq I_p^-+\tau \right\}$\;
    \lIf{$|\mathcal{C}_{p,j}| \ge 2$ and $|\NewerC_{p,j}|\geq 1$}{\Return \textbf{true}}
    \ForEach{$q \in \NewerC_{p,j}$}{
        \ForEach{$1 \le k \le 2^d$}{
          \lIf{$\DTRQ\left(p, \tau, \square^q \cap \square^p_{k}\right)\neq \emptyset$}{\Return \textbf{true}}
        }
    }
}
\Return \textbf{false}\;
\end{algorithm}

 \paragraph{Reporting Algorithm.} Suppose $\tau_{i+1}, \tau_i$ are the durability parameters of current and last query separately. For completeness, we set $\tau_0 = -\infty$. We distinguish the following two cases.
 
 \paragraph{Case 1: $\tau_{i+1} > \tau_i$.}
 We note that this case degenerates to the offline setting of reporting all $\tau_{i+1}$-durable triangles. We can simply issue a range query to $\mathcal{S}_\alpha$ to find $P_{\tau_{i+1}}$, i.e., the points that participate in at least one $\tau$-durable triangle, and then invoke Algorithm~\ref{alg:report-unit} for each point $p\in P_{\tau_{i+1}}$. Additionally, we reconstruct  $\S_\beta$ from scratch, by adding every point $p \in P_{\tau_{i+1}}$ into $\S_\beta$ with the value of $\beta^\tau_p$. 
 
 \paragraph{Case 2: $\tau_{i+1} < \tau_i$.}
 In this case, our target is to report all triangles in $T_{\tau_{i+1}} \setminus T_{\tau_i}$.
We first point out that any triangle in $T_{\tau_{i+1}} \setminus T_{\tau_i}$ falls into one of the following two cases: 
\begin{itemize}[leftmargin=*]
    \item there exists one point, say $p$, with $\tau_{i+1} \leq \alpha_p < \tau_{i}$.
    \item there exists one point, say $p$ with $\tau_{i+1} \le \beta^{\tau_{i}}_p < \tau_i < \alpha_p$.
\end{itemize}

To report triangles in the former case, we find all points $p$ with $\tau_{i+1} \leq \alpha_p < \tau_i$ from $\mathcal{S}_\alpha$.
Let $P_{\tau_{i+1}}^\alpha$ be the set of such points $p$. For each $p\in P_{\tau_{i+1}}^\alpha$, we invoke 
Algorithm~\ref{alg:report-unit} to report all $\tau_{i+1}$-durable triangles that contain $p$.

To report triangles in the latter case, we find all points $p$ with $\tau\leq \beta^{\tau_i}_p$ from $\S_\beta$. 
Let $P_{\tau_{i+1}}^\beta$ be the set of such points $p$.
For each such a point $p\in P_{\tau_{i+1}}^\beta$ we
invoke Algorithm~\ref{alg:report-delta} to report all triangles in $T_{\tau_{i+1}} \setminus T_{\tau_i}$ that contain each such a point $p$.
Notice that $P_{\tau_{i+1}}=P_{\tau_{i+1}}^\alpha\cup P_{\tau_{i+1}}^\beta$
%
%
%
\begin{algorithm}[t]
\caption{{\sc ReportDeltaTriangle-I}$(p, \DTR, \tau_i,  \tau_{i+1})$}
\label{alg:report-delta}
$\square^p \gets \{\square^p_1, \square^p_2, \ldots, \square^p_{2^d}\}$ as canonical unit squares of $p$\;
\ForEach{$j \in [2^d]$}{
    $\C_{p,j} \gets \DTRQ(p, \tau_{i+1}, \square^p_{j})$\;
    $\NewerC_{p,j} = 
    \left\{q \in \C_{p,j}: I^-_q \le I_p^-,  I_p^- + \tau_{i+1} \le I^+_q < I_p^- +{\tau_{i}}\right\}$\; $\OlderC_{p,j} = \left\{q \in \C_{p,j}: I_q^-\leq I_p^-, I_q^+\geq I_p^-+{\tau_{i}} \right\}$\;
    Report $(p,q,q')$ for every $(q,q') \in \NewerC_{p,j} \times \NewerC_{p,j}$ if $q \neq q'$\;
    Report $(p,q,q')$ for every $q \in \NewerC_{p,j}$ and $q' \in \OlderC_{p,j}$\;
    \ForEach{$q \in \NewerC_{p,j}$}{
        \ForEach{$k \in [2^d]$ with $k < j$}{
            $\mathcal{L}_{p,q} \gets \DTRQ(p, \tau_{i+1}, \square^q \cap \square^p_{k})\setminus\NewerC_{p,k}$\;
            Report $(p,q,q')$ for every $q' \in \mathcal{L}_{p,q}$\;
        }
        \ForEach{$k \in [2^d]$ with $k > j$}{
            $\mathcal{L}_{p,q} \gets \DTRQ(p, \tau_{i+1}, \square^q \cap \square^p_{k})$\;
            Report $(p,q,q')$ for every $q' \in \mathcal{L}_{p,q}$\;
        }
    }
}

\end{algorithm}
%
%
%
%
%

We construct a square with side-length $2$ and center $p$, and we split it into $2^d$ canonical unit squares $\square^p=\{\square^p_1, \ldots, \square^p_{2^d}\}$.
We also partition each $\C_{p,j}$ into two different subsets of points,  $\NewerC_{p,j}$ and $\OlderC_{p,j}$.
%
We further distinguish a $\tau_{i+1}$-durable but non ${\tau_{i}}$-durable triangle $(p,q,q')$ that contains $p$ into three types: (1) $q,q' \in \NewerC_{q,j}$; (2) $q \in \NewerC_{q,j}$, $q' \in \OlderC_{q,j}$; (3) $q \in \NewerC_{q,j}, q' \in \OlderC_{q,k}$ for $j \neq k$. We only need to report triangles in type (1) when $|\NewerC_{p,j}|>1$, i.e., there are multiple points from $\C_{p,j}$ with durability threshold in $[\tau_{i+1}, {\tau_{i}})$, hence every pair in the Cartesian product of $\NewerC_{p,j} \times \NewerC_{p,j}$ must form a $\tau_{i+1}$-durable but not ${\tau_{i}}$-durable triangle with $p$. We only need to report triangles in type (2) when $|\NewerC_{p,j}|\ge 1$ and $|\OlderC_{p,j}| \ge 1$, i.e., there exists a pair of points from $\C_{p,j}$ with durability threshold in $[\tau_{i+1}, {\tau_{i}})$ and $[{\tau_{i}}, +\infty)$ separately, hence every pair in the Cartesian product of $\NewerC_{p,j} \times \OlderC_{p,j}$ must form a $\tau_{i+1}$-durable but not ${\tau_{i}}$-durable triangle with $p$. To report triangles in type (3), for each point $q \in \NewerC_{p,j}$ returned, we issue another $\tau_{i+1}$-durable range query $\DTRQ\left(p, \tau_{i+1}, \square^q\cap \square^p_{k}\right)$ for every $k \neq j$, whose answer is exactly is the set of points that form a $\tau_{i+1}$-durable but not ${\tau_{i}}$-durable triangle with $p,q$.

After the end of the enumeration, we update $\S_\beta$.
By definition, it suffices to compute $\beta^{\tau_{i+1}}_p$ for every point $p \in P_{\tau_{i+1}}$. We distinguish two cases for each point $p \in P_{\tau_{i+1}}$: 
\begin{itemize}[leftmargin=*]
    \item if ${\tau_{i+1}} \leq \alpha_p < {\tau_{i}}$, we compute $\beta^{\tau_{i+1}}_p$ as described in the beginning of Section~\ref{sec:linfty-reporting}.
    \item otherwise, $\alpha_p > {\tau_{i}}$. We further distinguish two more cases:
    \begin{itemize}
        \item if $\beta^{\tau_{i}}_p < {\tau_{i+1}}$, we do not need to change the activation threshold for $p$, since $\beta^{\tau_{i}}_p = \beta^{\tau_{i+1}}_p$ in this case.
        \item otherwise, $\beta^{\tau_{i}}_p \geq {\tau_{i+1}}$, we compute $\beta^{\tau_{i+1}}_p$. 
    \end{itemize}
\end{itemize}


\paragraph{Correctness.} 
%
We show that the values $\beta_p^{\tau_{i+1}}$ are updated correctly.
Let $\tau'$ be the parameter in the binary search that we checked in Algorithm~\ref{alg:detect} 
Point $p$ can only form $\tau'$-durable triangles with points $q$ whose intervals $I_q$ intersect $I_p^-$ and either $I_q^+<I_p^-+\tau$ or $I_q^+\geq I_p^-+\tau$. $\bigcup_{j}\NewerC_{p,j}$ is the set of points satisfying the first inequality, and $\bigcup_j\OlderC_{p,j}$ the set of points in the second inequality. For every pair $q,s\in \OlderC_{p,j}$, we do not activate point $p$ with durability $\tau'$. If indeed $||q-s||_\infty<1$ and $q,s\in \OlderC_{p,j}$, then $(p,q,s)$ is a $\tau$-durable triangle. So our algorithm will not activate a point $p$ because of a previously reported $\tau$-durable triangle. It is also straightforward to see that $p$ should be activated at time $\tau'$ if there is a pair of points $q, s\in \NewerC_{p,j}\cup \OlderC_{p,j}$ such as either $q$ or $s$ belongs in $\NewerC_{p,j}$.

We next show the correctness of Algorithm~\ref{alg:report-delta}.
The case with ${\tau_{i+1}}>{\tau_{i}}$ follows from the offline setting.
The same arguments for Algorithm~\ref{alg:detect} apply here, such that we only consider points in $P_{\tau_{i+1}}$ that should be activated, and  Algorithm~\ref{alg:report-delta} only reports new $\tau_{i+1}$-durable triangles.

\paragraph{Time Complexity.}
For both cases ${\tau_{i+1}} > {\tau_{i}}$ and ${\tau_{i+1}}<{\tau_{i}}$, we spend $O\left(|T_{\tau_{i+1}} \setminus T_{\tau_i}| \cdot \log^{d+2} n \right)$ time to report all ${\tau_{i+1}}$-durable triangles, and $O(\log^{d+3} n)$ time to compute $\beta^{\tau_{i+1}}_p$ for each point $p \in P_{\tau_{i+1}}$. As defined, each point $p \in P_{\tau_{i+1}}$ must participate in at least one new ${\tau_{i+1}}$-durable triangle, which is also reported in this case. Hence, the overall time is bounded by $O\left(|T_{\tau_{i+1}} \setminus T_{\tau_i}| \cdot \log^{d+3} n\right)$. 

\begin{theorem}
\label{thm:online}

Given $(P,\dist, I)$, where $P$ is a set of $n$ points in $\Re^d$ and $\phi$ is the $\ell_\infty$ norm,
  a data structure of size $O(n\log^{d+2}n)$ can be constructed in $O(n\log^{d+3}n)$ time such that, 
\problemtwo\ can be solved in
  $O\left(|T_{\tau_{i+1}} \setminus T_{\tau_i}| \cdot \log^{d+3} n\right)$ time.
\end{theorem}

\section{Dynamic Setting}
\label{appndx:dynamic}
\newcommand{\problemDone}{\textsf{DynamicOffDurable}}
In this setting, we assume that we do not know the point set $P$ upfront.
We start with an empty point set $P'=\emptyset$, and some input parameters $\tau, \eps$. The goal is to construct a data structure such that, if all points are inserted and deleted according to their lifespans, it supports the following operations: i) if a point is deleted the data structure is updated efficiently, and ii) if a point $p$ is inserted, the data structure is updated efficiently and $\tau$-durable triangles (if any) of the form $(p,q,s)$ are reported such that $I_p^-\geq \max\{I_q^-,I_s^-\}$ along with some $\tau$-durable $\eps$-triangles that contain $p$. 
We call it the \problemDone\ problem.

We note that the data structure we need in this dynamic setting is a dynamic version of $\DTC$.
In particular, we slightly modify the standard techniques to convert our static data structure to a dynamic one with amortized update guarantees~\cite{overmars1981worst, erickson2011static, overmars1987design}.

\newcommand{\dynDTC}{\DTC^{\textsf{dyn}}}
We call the new dynamic data structure $\dynDTC$. Let $P'$ be the current instance of $O(n)$ ``active'' points.
$\dynDTC$ consists of $K=O(\log n)$ subsets of points $G_1, \ldots, G_K$ such that for each $i\leq K$, $G_i\subseteq P'$ and $\bigcup_{i\leq K} G_i=P'$. For each $i\leq K$, we have a static data structure $\DTC$, called $\DTC^i$, over $G_i$.
It holds that for each $i\leq K$, $|G_i|$ is either $2^i$ or $0$.
The total space of the data structure is $O(n\log n)$.

Assume that we have to remove $p\in P$. We identify the group $i$ such that $p\in G_i$. We also identify the leaf node $u$ of the cover tree $\mathcal{T}_i$ that $p$ belongs to. Both of these operations can easily be executed in $O(\log n)$ time with some auxiliary data structures. Then we traverse $\mathcal{T}_i$ from $u$ to the root removing $p$ from the linked interval trees. Notice that the structure of $\DTC^i$ does not change, instead only the information stored in the nodes of the interval trees containing $p$ are updated.

Next, assume that a new point $p$ is inserted at time $I_p^-$. We first place $p$ in a temporary min heap $H$ with value $I_p^-+\tau$ (the time instance that $p$ can participate in $\tau$-durable triangles). When we reach time $I_p^-+\tau$, we derive $p$ from $H$. At this point $p$ is an active point for time more than $\tau$. We  insert $p$ in $\dynDTC$ as follows. We find the smallest $i$ such that $G_i=\emptyset$. We move all points $A=\bigcup_{j<i}G_j$ into $G_i$ and construct $\DTC^i$ over $A\cup \{p\}$.
At this point, we also need run a query to find all $\tau$-durable triangles that contain $p$ (having $I_p^-$ as the largest left endpoint). We run the offline query $\DTCQ^i(p,\tau,\eps/2)$ in each $\DTC^i$ with $G_i\neq \emptyset$. As we had in the offline case, for each $i$ we get a set of $O(\eps^{-\ddim})$ canonical nodes of the cover tree. Each canonical node corresponds a ball of radius at most $\eps/4$. In the dynamic case there are in total $O(\eps^{-\ddim}\log n)$ canonical nodes, since there are $O(\log n)$ groups. In order to find the durable triangles, we run Algorithm~\ref{alg:offline-2} from the offline case considering $O(\eps^{-\ddim}\log n)$ canonical nodes instead of $O(\eps^{-\ddim})$. Hence, the running time is $O(\eps^{-\ddim}\log^3 n + \OUT_p)$, where $\OUT_p$ is the number of $\tau$-durable triangles anchored by $p$ (with $I_p^-$ being the largest left endpoint) along with a number of additional $\tau$-durable $\eps$-triangles anchored by $p$. Equivalently, we can argue that $\OUT_p$ is the $\tau$-durable triangles that $p$ participates in at the moment $I_p^-+\tau$ along with some additional $\tau$-durable $\eps$-triangles that $p$ participates in.
Finally, we re-construct $\dynDTC$ from scratch after $n/2$ updates.

Following the analysis in~\cite{overmars1981worst, erickson2011static, overmars1987design} and observing that each point can change at most $O(\log n)$ groups and the construction time of the offline $\DTC$ is $O(n\log^2 n)$, we have that the insertion of a point takes $O(\log^3 n)$ amortized time. The deletion takes $O(\log^2 n)$ time.

\begin{theorem}
\label{appndx:thm:dynDS}
Given $(P,\phi,I)$, $\tau > 0$, and $\eps > 0$,
    $\eps$-approximate \problemDone\ can be solved using a data structure of $O(n\log n)$ space, $O(\log^3 n)$ amortized update time, and $O(\eps^{O(-\ddim)}\log^3 n+\OUT_p)$ time to report all new $\tau$-durable triangles (along with some $\eps$-triangles) anchored by $p$, where $\OUT_p$ is the output size after inserting point $p$, where $n = |P|$, and $\ddim$ is the doubling dimension of $P$.
\end{theorem}

\section{Extensions}
\label{appndx:extensions}

We show how we can extend our results in any $\ell_\alpha$ norm and we describe how to report other patterns (except of triangles) of constant size. Furthermore, we show how to report all durable star patterns.
While we only describe the results in the offline setting, all of them can be extended to the online setting using the approach as shown in Section~\ref{sec:online}.

\subsection{$\ell_\alpha$ metric}

In order to find all $\tau$-durable triangles in any $\ell_\alpha$ metric we use a quadtree $\mathcal{T}$ instead of a cover tree over the input points $P$. Each node $u$ of the cover tree is associated with a square $\square_u$. Let $P_u-P\cap\square_u$. Given a point $p\in P$ we find a set of $O(\log n + \eps^{-d})$ canonical nodes $C$ in $\mathcal{T}$ such that for every $u\in C$, the diameter of $\square_u$ is at most $\eps/2$ and $||p-\square_u||_\alpha\leq 1+\eps/4$. Using the same procedure we followed for constant doubling dimensions over the canonical subsets $\C_p$, we get the following result.

\begin{theorem}
Given $(P,\phi,I)$, $\eps>0$, and $\tau>0$, where $P$ is a set of $n$ points in $\Re^d$ and $\phi$ is any $\ell_{\alpha}$ norm,
the $\eps$-approximate \problemone\ can be solved in
 $O\left(n(\eps^{-d}\log n + \eps^{-2\cdot d}+\log^2 n) + \OUT\right)$ time, 
    where $\OUT$ is the number of triangles reported, satisfying $|T_\tau|\leq \OUT\leq |T^{\eps}_\tau|$.
\end{theorem}

\subsection{Other patterns}
In this subsection we show how we can extend the offline algorithm to report i) durable cliques of constant size, ii) durable paths of constant size, and iii) durable $k$-star patters. The algorithms can also be extended to handle incremental queries, similarly to $tau$-durable triangles.

\mparagraph{Cliques}
Let $S\subseteq P$ be a subset of points.
$S$ is a $\tau$-durable $m$-clique if i) $|S|=m$, ii) for every pair $p,q\in S\times S$, $\dist(p,q)\leq 1$, and iii) $|\cap_{p\in S}I_p|\geq \tau$.
Similarly, $S$ is called a $\tau$-durable $\eps$-$m$-clique if i, iii remain the same and for every pair $p,q\in S\times S$, $\dist(p,q)\leq 1+\eps$.

We consider that $m=O(1)$.
The algorithm to report all $\tau$-durable $m$-cliques is similar to Algorithm~\ref{alg:offline-2}. The only difference is that instead of considering all pairs $\C_{p,i}$, $\C_{p,j}$ of the nodes in the cover tree, we consider all possible subsets of size $m$. Let $\C_{p,j_1},\ldots, \C_{p,j_m}$ be the family of $m$ subsets. If all pairwise distances among the representative points are at most $1+\eps/2$ then we report all $m$-cliques $\C_{p,j_1}\times\ldots\times \C_{p,j_m}$. The correctness follows straightforwardly from Section~\ref{sec:offline}. In particular we report all $\tau$-durable $m$-cliques and we might also report a few $\tau$-durable $\eps$-$m$-cliques. The running time is also asymptotically the same with Algorithm~\ref{alg:offline-2}.

\mparagraph{Paths}
Let $S\subseteq P$ be a subset of points.
$S$ is a $\tau$-durable $m$-path if i) $|S|=m$, ii) there is an ordering of the points such that the distance of two consecutive points is at most $1$, and iii) $|\cap_{p\in S}I_p|\geq \tau$.
Similarly, $S$ is called a $\tau$-durable $\eps$-$m$-path if i), iii) remain the same, and the distance between consecutive points is at most $1+\eps$.

We consider that $m=O(1)$.
The algorithm to report all $\tau$-durable $m$-cliques is similar to Algorithm~\ref{alg:offline-2}. The only difference is that instead of considering all pairs $\C_{p,i}$, $\C_{p,j}$ of the nodes in the cover tree, we consider all possible subsets of size $m$. Let $\C_{p,j_1},\ldots, \C_{p,j_m}$ be the family of $m$ subsets. We try all possible $O(m!)=O(1)$ orderings and we check if we find an ordering 
$\C_{p,j_1},\ldots, \C_{p,j_m}$ such that $\dist(\rep_{j_1},\rep_{j_2})\leq 1+\eps/2$, $\dist(\rep_{j_2},\rep_{j_3})\leq 1+\eps/2$, \ldots $\dist(\rep_{j_{m-1}},\rep_{j_m})\leq 1+\eps/2$. If this is true then we report all $m$-paths $\C_{p,j_1}\times\ldots\times \C_{p,j_m}$. The correctness follows straightforwardly from Section~\ref{sec:offline}.
In particular we report all $\tau$-durable $m$-paths and we might also report a few $\tau$-durable $\eps$-$m$-paths.
The running time is also asymptotically the same with Algorithm~\ref{alg:offline-2}.

\mparagraph{$k$-star patterns}
Let $S\subseteq P$ be a subset of points.
$S$ is a $\tau$-durable $m$-star if i) $|S|=m$, ii) there is a central point $p\in S$ such that $\dist(p,q)\leq 1$ for every other $q\in S$, and iii) $|\cap_{p\in S}I_p|\geq \tau$.
Similarly, $S$ is called a $\tau$-durable $\eps$-$m$-star if i), iii) remain the same, and the distance between $p$ and any other point in $S$ is at most $1+\eps$.

We consider that $m=O(1)$.
The algorithm to report all $\tau$-durable $m$-cliques is similar to Algorithm~\ref{alg:offline-2}.
For each point $p\in P$, we run a query $\DTCQ(p,\tau, \eps/2)$, but instead of querying points within distance $1$ from $p$, i.e., points in ball $\mathcal{B}(p,1)$, we query pints within distance $2$ from $p$, i.e., points in ball $\mathcal{B}(p,2)$. We need that change because $p$ might belong to an $m$-star pattern $S$ having the largest left endpoint $I_p^-\geq \max_{q\in S}I_q^-$, while not being the central point. In this case, it is always true that $S\subseteq P\cap \mathcal{B}(p,2)$. Hence, we get $\C_p=\{\C_{p,1},\ldots, \C_{p,k}\}$, for $k=O(\eps^{-O(\ddim)})$ canonical nodes of the cover tree that approximately cover $\mathcal{B}(p,2)$. Then we visit each node $\C_{p,j}$. We initialize a counter $c=0$. For every other node $\C_{p,h}\in \C_p$ we check whether $\dist(\rep_j,\rep_h)\leq 1+\eps/2$. If yes then we update $c=c+|\C_{p,h}|$. In the end, if $c>m$ there exist $|\C_{p,j}|$, $m$-star patterns to report. Hence, for each point $q\in \C_{p,j}$ we report $q$ as the central point and then we visit all nodes $\C_{p,h}$ with $\dist(\rep_j,\rep_h)\leq 1+\eps/2$ to report all points in $C_{p,h}$.
The correctness follows straightforwardly from Section~\ref{sec:offline} and the fact that for each $p$ we report all $m$-star patterns that $p$ belongs to (not necessarily as the central point) having the maximum left endpoint on its corresponding temporal interval. In particular we report all $\tau$-durable $m$-star patterns and we might also report a few $\tau$-durable $\eps$-$m$-star patterns. The running time is also asymptotically the same with Algorithm~\ref{alg:offline-2}.


\section{Missing Material in Section~\ref{sec:aggregate}}
\label{appendix:aggregate}


\subsection{UNION}
\begin{algorithm}
\caption{{\sc ReportUNIONPair}$(\DTCUNION, p, \tau, \eps,\kappa)$}
\label{alg:union2}

$\C_{p}: \{\C_{p,1}, \C_{p,2}, \cdots, \C_{p,k}\} \gets \DTCQ(p,\tau,\eps/2)$,
    with $\rep_i$ denoting the representative point of the cover tree node for $\C_{p,i}$,
    and $\ITUNION_{p,i}$ denoting the annotated interval tree for this cover tree node\;
\ForEach{$j \in [k]$}{
    \ForEach{$q \in \C_{p,j}$ in descending order of $I_q^+$}{
        $I'\gets I_q\cap I_q$\;
        $\bar{I}=\textsc{MaxIntersection}(I')$\;
        $H\gets \textsf{newHeap}(\{(\bar{I},I', |I\cap I'|)\})$\;
        $t\gets 0$\;
        \For{$h=1 \ldots \kappa$}{
            $(I_x,I_y,|I_x\cap I_y|)=H.removeTop()$\;
            $t\gets t+|I_x\cap I_y|$\;
            \ForEach{$I_z\in I_y\setminus I_x$}{
                $\bar{I}=\textsc{MaxIntersection}(I_z)$\;
                $H.\textsf{insert}(\bar{I}, I_z, |\bar{I}\cap I_z|)$\;
            }
        }
        \If{$t\geq (1-1/e)\tau$}{
            \KwSty{report} $(p,q)$\;
        }\Else{
            \KwSty{break}\;
        }

    }
}

\KwSty{Subroutine} $\textsc{MaxIntersection}(I_{in})$ \Begin{
            $\mu \gets -1$\;
            $\hat{I}\gets \emptyset$\;
            \ForEach{$i \in [k]$}{
                \If{$\dist(\rep_i, \rep_j) \le 1 +\frac{\eps}{2}$}{
                    $\hat{I}_i \gets \textsc{ComputeMaxUnionD}(\ITUNION_{p,i}, I_{in})$\;
                }
                 \If{$|\hat{I}_i\cap I_{in}|>\mu$}{
                $\mu=|\hat{I}_i\cap I_{in}|$\;
                $\hat{I}=\hat{I}_i$\;
                }
            }
            \Return ($\hat{I}$)\;
}
\end{algorithm}

In Algorithm~\ref{alg:union2} we show the algorithm for finding all $(\tau,\kappa)$-\UNION\ durable pairs. The high level idea of the algorithm follows from Algorithm~\ref{alg:sum}.

Instead of $\ITSUM$, we have the primitive data structure $\ITUNION$. Given a query interval $I_{in}$ the goal is to find the interval $\hat{I}\in \mathcal{I}$ such that $|\hat{I}\cap I_{in}|$ is maximized, where $\mathcal{I}=\{I_p\mid p\in P\}$. This can be found using a variant of the interval tree $\ITUNION$ as follows: First, among all intervals that intersect $I_{in}^-$ it finds the interval $\hat{I}_a\in\mathcal{I}$ with the largest right endpoint. Second, among all intervals that intersect $I_{in}^+$ it finds the interval $\hat{I}_b\in\mathcal{I}$ with the smallest left endpoint. Third, among all intervals that lie completely inside $I_{in}$ it finds the longest interval $\hat{I}_c\in\mathcal{I}$. In the end, we return $\hat{I}\in\{\hat{I}_a, \hat{I}_b, \hat{I}_c\}$ with the longest intersection  $|\hat{I}\cap I_{in}|$. Similarly to $\ITSUM$, the data structure $\ITUNION$ has space $O(n\log n)$, it can be constructed in $O(n\log^2 n)$ time, and given a query interval $I_{in}$, it returns $\hat{I}$ in $O(\log^2 n)$ time.

Using $\ITUNION$, we construct $\DTCUNION$ similarly to $\DTCSUM$. The only difference is that for each node in the cover tree there exist an $\ITUNION$ data structure (instead of $\ITSUM$). The procedure $\textsc{ComputeMaxUnionD}   
(\ITUNION_{p,i}, I_{in})$ returns the interval $\hat{I}_i$ that has the largest intersection with $I_{in}$ and its corresponding point in $P$ lies in $\C_{p,i}$.

Next, we describe 
the subroutine $\textsc{MaxIntersection}(I_{in})$. It simply considers all canonical nodes in $\C_p$
running $\textsc{ComputeMaxUnionD}(\ITUNION_{p,i}, I_{in})$ for each $\C_{p,i}\in \C_p$ such that $\dist(\rep_i,\rep_j)\leq 1+\eps/2$. Hence $\textsc{MaxIntersection}(I_{in})$ visits all canonical nodes that are close to both $p$ and $q$, and among these nodes $\C_{p,i}$, it returns the interval $\hat{I}=\argmax_{\hat{I}_i}|\hat{I}_i\cap I_{in}|$, i.e., the interval with the largest intersection with $I_{in}$. We also notice that $\textsc{ComputeMaxUnionD}$ can be easily modified so that we always skip $I_p$ and $I_q$ from the procedure of finding the interval in $\mathcal{I}$ with the largest intersection with $I_{in}$.
From the proofs in the previous sections we have that $\dist(p,q)\leq 1+\eps/2$ while the witness set has distance at most $1+\eps$ from both $p$ and $q$. Overall, $\textsc{MaxIntersection}(I_{in})$ finds an interval $\hat{I}=I_{s}$ such that $$|I_s\cap I_{in}|\geq \max_{s'\in P: \dist(p,s'), \dist(q,s')\leq 1}|I_{s'}\cap I_{in}|,$$ and $$\dist(p,s), \dist(q,s)\leq 1+\eps.$$

Finally, we describe {\sc ReportUNIONPair}$(\DTCUNION, p, \tau, \eps,\kappa)$. After finding $\C_p$ and for each $\C_{p,j}$ we visit $q\in \C_{p,j}$ in descending order of $I_q^+$ as we did in Algorithm~\ref{alg:sum}. For each pair $(p,q)$ we check, we run the greedy algorithm for the max $k$-coverage problem. First, we find the interval $\bar{I}\in \I$ that has the largest intersection with $I'=I_p\cap _q$. We create a max heap $H$ and we insert the pair $(\bar{I}, I')$ with value $|\bar{I}\cap I'|$. Then the algorithm proceeds in $\kappa$ iterations. In each iteration, it finds the pair $(I_x,I_y)$ in the max heap $H$ with the maximum $|I_x\cap I_y|$. $I_x$ is an interval from $\mathcal{I}$, while $I_y$ is an uncovered segment of $I'$. Hence, in each iteration, it finds the interval $I_x$ that covers the largest uncovered area of $I'$. Then we add $|I_x\cap I_y|$ in variable $t$ that maintains the overall union the algorithm has computed. An interval $I_x$ might split $I_y$ into two smaller uncovered segments or into one smaller uncovered segment of $I'$. In each case, $I_z\subseteq I_y$ represents one uncovered segment created after adding $I_x$. We run $\textsc{MaxIntersection}(I_z)$ and we find the interval $\bar{I}\in \mathcal{I}$ that covers the largest portion of $I_z$ and we insert the pair $(\bar{I}, I_z)$ in max heap $H$ with value $|\bar{I}\cap I_z|$. We repeat the same procedure for $\kappa$ iteration. In the end we check whether $t\geq (1-1/e)\tau$. If yes, we report the pair $(p,q)$, otherwise we skip $\C_{p,j}$ and continue with the next canonical node. Overall, this algorithm gives an implementation of the greedy algorithm for the max $\kappa$-coverage problem in our setting, using efficient data structure to accelerate the running time.

Putting everything together, the correctness of this algorithm follows by the correctness of the greedy algorithm for the max $\kappa$-coverage problem, the correctness of $\textsc{MaxIntersection}(I_{in})$, and the correctness of Algorithm~\ref{alg:sum}.
For a pair $(p,q)$ if we find that $t\geq(1-1/e)\tau$, then $(p,q)$ is definitely an $((1-1/e)\tau,\kappa)$-\UNION\ $\eps$-pair.
As we argued in Algorithm~\ref{alg:sum}, assume that for a $q\in \C_{p,j}$ we find that $t<(1-1/e)\tau$. Then it is safe to skip $\C_{p,j}$ because there is no other $(\tau,\kappa)$-\UNION\ durable pair to report. 
Notice that the approximation factor for the greedy algorithm is $1-1/e$, so if the greedy implementation returns $t<(1-1/e)\tau$, we are sure that the pair $(p,q)$ is not $(\tau,\kappa)$-\UNION\ durable. Hence, any other $w\in\C_{p,j}$ with $I_w^+<I_q^+$ will also not be $(\tau,\kappa)$-\UNION\ durable.
In any case, our algorithm returns all $(\tau,\kappa)$-\UNION\ durable pairs and might return some additional $((1-1/e)\tau,\kappa)$-\UNION\ durable $\eps$ pairs. Hence, it holds that $|K_{\tau,\kappa}| \leq \OUT \leq |K_{(1-1/e)\tau, \kappa}^{\eps}|$.

Next, we analyze the running time of our algorithm. As pointed out, $\ITUNION$ is constructed in $O(n\log^2 n)$ time and it finds the interval that covers the largest uncovered area of a query interval in $O(\log^2 n)$ time.
Hence, $\DTCUNION$ is constructed in $O(n\log^3 n)$ time.
The subroutine $\textsc{MaxIntersection}(I_{in})$ calls $\textsc{ComputeMaxUnionD}(\ITUNION_{p,i}, I_{in})$, $O(\eps^{-O(\ddim)})$ times. For each pair $(p,q)$ we check, the subroutine $\textsc{MaxIntersection}(I_{in})$ is called $O(\kappa)$ times, while all update operations in the max heap $H$ takes $O(\kappa\log n)$ time. For each point $p$ we might check at most $O(\eps^{-O(\ddim)})$ pairs that are not reported, one for each canonical node in $\C_p$.
Overall, Algorithm~\ref{alg:union2} runs in $O(n\log^3 n+(n+\OUT)\eps^{-O(\ddim)}\kappa\log^2 n)$ time.

\end{document}